\RequirePackage[l2tabu]{nag}
\documentclass{coltstyle-arxiv} % change [12pt] to [anon,12pt] before submission
\pdfoutput=1

\usepackage{tikz}

\usepackage{algorithm}
\usepackage[noend]{algorithmic}
\usepackage{amsfonts}
\usepackage{comment}

\newtheorem{fact}{Fact}
\newtheorem{claim}{Claim}

\usepackage{enumerate}

\newcommand{\ts}{\textstyle}
\allowdisplaybreaks

\DeclareMathOperator*{\argmax}{arg\,max}
\DeclareMathOperator*{\argmin}{arg\,min}

\DeclareMathOperator{\cost}{cost}

\newcommand{\cramped}{\vspace{-\topsep}\itemsep0pt}
\newcommand{\proofof}[1]{{\bfseries\upshape of #1\ }}

\title{Scenario Submodular Cover}

\coltauthor{%
  \Name{Nathaniel Grammel}\thanks{Partially Supported by NSF Grant 1217968} \Email{ngrammel@nyu.edu}\\
  \addr Department of Computer Science and Engineering\\
  NYU Tandon School of Engineering\\
  Brooklyn, NY 11201
  \AND
  \Name{Lisa Hellerstein}\footnotemark[1] \Email{lisa.hellerstein@nyu.edu}\\
  \addr Department of Computer Science and Engineering\\
  NYU Tandon School of Engineering\\
  Brooklyn, NY 11201
  \AND
  \Name{Devorah Kletenik}\footnotemark[1] \Email{kletenik@sci.brooklyn.cuny.edu}\\
  \addr Department of Computer and Information Science \\
  Brooklyn College, City University of New York \\
  2900 Bedford Avenue\\
  Brooklyn, NY 11210
  \AND
  \Name{Patrick Lin}\footnotemark[1] \Email{plin15@illinois.edu}\\
  \addr Department of Computer Science\\
  University of Illinois at Urbana-Champaign\\
  Urbana, IL
}

\begin{document}

\maketitle

\begin{abstract}
  Many problems in Machine Learning can be modeled as submodular
  optimization problems.  Recent work has focused on stochastic or
  adaptive versions of these problems.  We consider the Scenario
  Submodular Cover problem, which is a counterpart to the Stochastic
  Submodular Cover problem studied by~\cite{golovinKrause}.  In
  Scenario Submodular Cover, the goal is to produce a cover with
  minimum expected cost, where the expectation is with respect to an
  empirical joint distribution, given as input by a weighted sample of
  realizations.  In contrast, in Stochastic Submodular Cover, the
  variables of the input distribution are assumed to be independent,
  and the distribution of each variable is given as input.  Building
  on algorithms developed by~\cite{cicaleseLaberSaettler}
  and~\cite{golovinKrause} for related problems, we give two
  approximation algorithms for Scenario Submodular Cover over discrete
  distributions.  The first achieves an approximation factor of
  $O(\log Qm)$, where $m$ is the size of the sample and $Q$ is the
  goal utility.  The second, simpler algorithm achieves an
  approximation bound of $O(\log QW)$, where $Q$ is the goal utility
  and $W$ is the sum of the integer weights.  (Both bounds assume an
  integer-valued utility function.)
% Our algorithms and their analysis are based on previous work of
% ~\cite{cicaleseLaberSaettler}, 
% We also present another, simpler, algorithm which is a  modified version of
% the Adaptive Greedy algorithm of~\cite{golovinKrause}. It achieves an approximation factor of $O(\log QW)$
Our results yield approximation bounds for other problems involving non-independent distributions
that are explicitly specified by their support.
\end{abstract}

%\begin{keywords}
%submodularity, optimization, approximation algorithms, decision trees
%\end{keywords}

\section{Introduction}
Many problems in Machine Learning can be modeled as submodular optimization problems.
Recent work has focused on stochastic or adaptive versions of submodular optimization problems, which 
reflect the need to make sequential decisions when outcomes are uncertain.

The Submodular Cover problem generalizes the classical NP-complete Set Cover
problem and is a fundamental problem in submodular optimization.
Adaptive versions of this problem have applications to a variety of machine learning problems that require building a decision tree, where the goal is to minimize expected cost. 
Examples include problems of entity identification (exact learning with membership queries),
classification (equivalence class determination), and decision region identification 
(cf.~\cite{golovinKrause, golovinKrauseRayNIPS, bellalaScottIT2012,javdaniDecision}). 
% These problems arise in areas as diverse as online recommender systems, automated medical diagnostics, and robotic manipulation. 
Other applications include reducing prediction costs for learned Boolean classifiers,
when there are costs for determining attribute values (\cite{deshpandeHellersteinKletenik}).

Previous work on the \emph{Stochastic} Submodular Cover problem assumes that the 
variables of the
input probability
distribution are independent.  Optimization is performed with respect to this distribution.
We consider a new version of the problem
that we call Scenario Submodular Cover, that removes the independence assumption.
In this problem, optimization is performed with respect to an input distribution
that is given explicitly by its support (with associated probability weights).
We give approximation algorithms solving the Scenario Submodular Cover problem over discrete
distributions.

Before describing our contributions in more detail, we give some background.
In generic terms,
an adaptive submodular cover problem is a sequential
decision problem where we must choose items one by one from an item set
$N=\{1,\dots,n\}$. Each item has an initially unknown state, which is
a member of a finite state set $\Gamma$.  The state of an item is
revealed only after we have chosen the item.  We represent a subset
$S$ of items and their states by a vector
$x \in (\Gamma \cup \{*\})^{n}$ where $x_i = *$ if $i \not\in S$, and
$x_i$ is the state of item $i$ otherwise.  We are given a monotone,
submodular
utility function
$g\colon(\Gamma \cup \{*\})^{n}\rightarrow \mathbb{Z}_{\geq 0}$.
It
assigns a non-negative integer value to each subset of the items and
the value can depend on the states of the items.\footnote{The definitions of the terms ``monotone'' and ``submodular,'' for state-dependent
utility functions, has not been standardized.  
We define these terms in Section~\ref{sec:defs}. 
In the terminology used by Golovin and
  Krause~\cite{golovinKrause}, $g$ is {\em pointwise} monotone and
  pointwise submodular.}
There is a non-negative goal utility value $Q$, such that
$g(a) = Q$ for all $a \in \Gamma^n$.  
There is a
cost associated with choosing each item, which we are given.
In distributional settings, we are also given the joint distribution of the item states.
We must continue
choosing items until their utility value is equal to the goal utility, $Q$.
The problem is to determine the adaptive order in which to choose the items so
as to minimize expected cost (in distributional settings) or worst-case cost (in adversarial settings).

Stochastic Submodular Cover is an adaptive submodular cover problem,  in a distributional setting.
In this problem, 
the state of each item is a random variable, and these variables are assumed to be independent.
The distributions of the variables are given as input.
Golovin and Krause introduced a simple
greedy algorithm for this problem, called Adaptive Greedy, that
achieves an approximation factor of $O(\log Q)$.  A dual greedy algorithm
for the problem, called
Adaptive Dual Greedy,
was presented and analyzed by~\cite{deshpandeHellersteinKletenik}. These greedy algorithms have
been useful in solving other stochastic optimization problems, which can
be reduced to Stochastic Submodular Cover through the construction of
appropriate utility functions (e.g.,
~\cite{javdaniDecision,javdaniVOI,deshpandeHellersteinKletenik,golovinKrauseRayNIPS}).

% 
% In practice, the independence assumptions used in solving stochastic
% optimization problems can be unacceptably strong.  We are therefore
% interested in versions of the problems which allow for correlation
% between the states of different items.

The problem we study in this paper, \emph{Scenario Submodular Cover}
(Scenario SC), is also a distributional, adaptive submodular cover
problem.  The distribution is given by a weighted sample, which is
provided as part of the input to the problem.  Each element of the
sample is a vector in $\Gamma^n$, representing an assignment of states
to the items in $N$.  Associated with each assignment is a positive
integer weight.  The sample and its weights define a joint
distribution on $\Gamma^n$, where the probability of a vector $\gamma$
in the sample is proportional to its weight.  (The probability of a
vector in $\Gamma^n$ that is not in the sample is 0.)  As in
Stochastic Submodular Cover, the problem is to choose the items and
achieve utility $Q$, in a way that minimizes the expected cost
incurred.  However, because many of the proofs of results for the
Stochastic Submodular Cover problem rely on the independence
assumption, the proofs do not apply to the Scenario SC problem.

\subsubsection*{Results}

We present an approximation algorithm for the Scenario SC
problem that we call {\em Mixed Greedy}. It uses two
different greedy criteria.  It is a generalization of an algorithm by~\cite{cicaleseLaberSaettler}
for the Equivalence Class Determination problem
(which has also been called the Group Identification problem and the
Discrete Function Evaluation problem).

The approximation factor achieved by Mixed Greedy for the Scenario
SC problem is $O\left(\frac{1}{\rho}\log Q\right)$, where $\rho$ is a quantity
that depends on the utility function $g$.  In the case of the utility
function constructed for the Equivalence Class Determination Problem,
$\rho$ is constant, but this is not true in general.  

We describe a
modified version of Mixed Greedy that we call {\em Scenario Mixed
Greedy}.  It works by first constructing a new monotone, submodular
utility function $g_S$ from $g$ and the sample, for which $\rho$ is
constant.  It then runs Mixed Greedy on $g_S$ with goal value $Qm$,
where $m$ is the size of the sample.  We show that Scenario Mixed Greedy
achieves an $O(\log Qm)$ approximation factor for any
Scenario SC problem.

Mixed Greedy is very similar to the algorithm of Cicalese et al., and
we use the same basic analysis.  However, at the heart of their
analysis is a technical lemma with a lengthy proof bounding a quantity
that they call the ``sepcost''.  The proof applies only to the
particular utility function used in the Equivalence Class
Determination problem.  We replace this proof with an entirely
different proof that applies to the general Scenario SC problem.
Our proof is based on the work of \cite{streeterGolovin} for the Min-Sum Submodular
Cover problem.

In addition to presenting and analyzing Mixed Greedy, we also present
another algorithm for the Scenario SC problem that we call 
\emph{Scenario Adaptive Greedy}.  It is a
modified version of the Adaptive Greedy algorithm of Golovin and
Krause.  Scenario Adaptive
Greedy is simpler and more efficient than Mixed Greedy, and is
therefore likely to be more useful in practice.  However, the
approximation bound proved by Golovin and Krause for Adaptive Greedy
depends on the assumption that $g$ and the distribution defined by the
sample weights jointly satisfy the {\em adaptive submodularity}
property.
% Roughly speaking, adaptive submodularity holds if $g$ is submodular
% in expectation with respect to the distribution.
This is not the case for general instances of the Scenario SC
problem.  We extend the approach used in constructing $g_S$ to give a
simple, generic method for constructing a modified utility function
$g_W$, with goal utility $QW$, from $g$, which incorporates the
weights on the sample.  We prove that utility function $g_W$ and the
distribution defined by the sample weights jointly satisfy adaptive
submodularity.  This allows us to apply the Adaptive Greedy algorithm,
and to achieve an approximation bound of $O(\log QW)$ for the
Scenario SC problem, where $W$ is the sum of the weights.

Our constructions of $g_S$ and $g_W$ are similar to constructions used
in previous work on Equivalence Class Determination and related
problems (cf.~\cite{golovinKrauseRayNIPS,bellalaScottIT2012,javdaniVOI,javdaniVOIlong}).
Our proof of adaptive submodularity uses the same basic
approach as used in previous
work (see, e.g.,~\cite{golovinKrauseRayNIPS,javdaniVOI,javdaniVOIlong}), namely
showing that the value of a certain function is non-decreasing along a
path between two points; however, we are addressing a more general
problem and the details of our proof are different.

% Approximation bounds for other problems follow easily from our
% results.  For example, Javdani et al.\ gave an
% $O(k \log \frac{W}{w_{min}})$ approximation algorithm for the
% Decision Region Identification problem, which is a generalization of
% the Equivalence Class Determination problem.  Here $W$ is the sum of
% the weights on the elements in the sample, $w_{min}$ is the smallest
% weight, and $k$ is another parameter associated with the problem.
% Their algorithm runs Adaptive Greedy.  Replacing it with Mixed
% Greedy immediately gives an approximation factor of $O(k\log m))$,
% where $m \leq \frac{W}{w_{min}}$ is the size of the sample.  We can
% also obtain results on Sample-Based versions of Stochastic Boolean
% Function evaluation problems and related sequential testing
% problems, such as the ones studied by Deshpande et al.~\cite{}.  We
% discuss this more fully in the body of the paper.

We believe that our work on Adaptive Greedy should make it easier to
develop efficient approximation algorithms for sample-based problems
in the future.  Previously, using ordinary Adaptive Greedy to solve a
sample-based problem involved the construction of a utility function
$g$, and a proof that $g$, together with the distribution on the
weighted sample, was adaptive submodular.  The proof was usually the most
technically difficult part of the work (see, e.g.,~\cite{golovinKrauseRayNIPS,bellalaScottIT2012,javdaniDecision,javdaniVOIlong}).
Our construction of $g_W$, and our proof of adaptive submodularity,
make it possible to achieve an approximation bound using Adaptive
Greedy after proving only submodularity of a constructed $g$, rather
than adaptive submodularity of $g$ and the distribution.  Proofs of
submodularity are generally easier because they do not involve
distributions and expected values.  Also, the standard OR construction
described in Section~\ref{sec:defs} preserves submodularity, while it
does not preserve Adaptive Submodularity (\cite{javdaniVOI}).

% A potential disadvantage of our approach is that the $O(\log QW)$
% and $O(\log Qm)$ bounds have a dependence on the size and weights of
% the sample.  However, we note that in previous work on Equivalence
% Class Determination and Decision Region Identification, when an
% $O(\log Q)$ bound was attained, $Q$ itself had had a polynomial
% dependence on $m$ or $W$.  The constructions of our modified utility
% functions $g_S$ and $g_W$ are similar to the constructions used in
% that previous work, but are generic.

Given a monotone, submodular $g$ with goal value $Q$, we can 
use the algorithms in this paper
to immediately obtain three approximation results for the associated
Scenario SC problem:
running Mixed Greedy with $g$ yields an
$O\left(\frac{1}{\rho}\log Q\right)$ approximation, running Mixed Greedy with
$g_S$ yields an $O(\log Qm)$ approximation, and running Adaptive
Greedy with $g_W$ yields an $O(\log QW)$ approximation.  
By the results of~\cite{golovinKrause}, 
running Adaptive Greedy with $g$ yields an
$O(\log Q)$ approximation for the associated Stochastic SC problem.

\subsubsection*{Applications}
Our results on Mixed Greedy yield approximation bounds for other
problems.  For example, we can easily obtain a new bound for the
Decision Region Identification problem studied by~\cite{javdaniDecision},
which is an extension of the Equivalence
Class Determination problem.  Javdani et al.~construct a utility
function whose value corresponds to a weighted sum of the hyperedges
cut in a certain hypergraph.  We can define a corresponding utility
function whose value is the \emph{number} of hyperedges cut.  This utility
function is clearly monotone and submodular.  Using Mixed Greedy with
this utility function yields an approximation bound of $O(k \log m)$,
where $k$ is a parameter associated with the problem, and $m$ is the
size of the input sample for this problem.  In contrast, the bound
achieved by Javdani et al.  is $O\left(k\log\left(\frac{W}{w_{min}}\right)\right)$, where
$w_{min}$ is the minimum weight on a assignment in the sample.

We can apply our greedy algorithms to Scenario BFE (Boolean
Function Evaluation) problems, which we introduce here.  These
problems are a counterpart to the Stochastic BFE problems\footnote{In the
  Operations Research literature, Stochastic Function Evaluation is
  often called Sequential Testing or Sequential Diagnosis.} 
that have been studied in AI, operations research, and in the context of learning with attribute costs
(see
e.g.,~\cite{unluyurtreview,deshpandeHellersteinKletenik,kaplanMansour-Stoc05}).
In
a Scenario BFE problem, we are given a Boolean function $f$.  For
each $i \in \{1,\ldots,n\}$, we are also given a cost $c_i > 0$
associated with obtaining the value of the $i$th bit of an 
initially unknown assignment
$a \in \{0,1\}^n$.  Finally, we are given a weighted sample
$S \subseteq \{0,1\}^n$.  The problem is to compute a (possibly
implicit) decision tree computing $f$, such that the expected cost of
evaluating $f$ on $a \in \{0,1\}^n$, using the tree, is minimized.
The expectation is with respect to the distribution defined by the
sample weights.

\cite{deshpandeHellersteinKletenik} gave 
approximation algorithms for some Stochastic BFE
problems that work by constructing an appropriate monotone, submodular utility function
$g$ and running Adaptive Greedy.  
By substituting the sample-based algorithms in this paper in place of Adaptive Greedy,
we obtain approximation results for analogous Scenario BFE problems.
For example, using Mixed Greedy, we can show that the
Scenario BFE problem for $k$-of-$n$ functions has an approximation
algorithm achieving a factor of $O(k \log n)$ approximation,
independent of the size of the sample.  Details are in
Appendix~\ref{sec:rhofork}.  Bounds for other functions follow easily
using Scenario Mixed Greedy and Scenario Adaptive Greedy.  For
example, \cite{deshpandeHellersteinKletenik} presented
an algorithm achieving an $O(\log t)$ approximation for the Stochastic
BFE problem for evaluating decision trees of size $t$.  Substituting
Scenario Mixed Greedy for Adaptive Greedy in this algorithm yields
an $O(\log tm)$ approximation for the associated Scenario BFE
problem.

% We achieve this bound by running Sample-Based Adaptive Greedy
% on the same $g$, with respect to the sample distribution, instead of
% running Adaptive Greedy with respect to the distribution defined by
% the $p_i$ in the SBFE problem.  Running Sample-Based Mixed Greedy
% instead would yield an approximation bound of $O(\log tm)$, where $m$
% is the size of the sample.  

We note that our Scenario BFE problem differs from the 
function evaluation problem 
by~\cite{cicaleseLaberSaettler}.
In their problem, the
computed decision tree need only compute $f$ correctly on assignments
$a \in \{0,1\}^n$ that are in the sample, while ours needs to compute $f$ correctly on all
$a \in \{0,1\}^n$.  To see the difference, consider the problem of
evaluating the Boolean OR function, for a sample $S$ consisting of only
$a \in \{0,1\}^n$ with at least one 1.  If the tree only has to
be correct on $a \in S$, a one-node decision tree that immediately
outputs $1$ is valid, even though it does not compute the OR function.
Also, in Scenario BFE we assume that the function $f$ is given with the sample, and we
consider particular types of functions $f$.

\subsubsection*{Organization}
We begin with definitions in Section~\ref{sec:defs}. 
In Section~\ref{sec:MG}, we present
the overview of the Mixed Greedy algorithm. 
Finally, we present
Scenario Mixed Greedy
in Section~\ref{sec:SMG}, followed by Scenario Adaptive Greedy in Section~\ref{sec:SAG}.

\section{Definitions}
\label{sec:defs}
Let $N = \{1,\dots, n\}$ be the set of \textit{items} and $\Gamma$ be
a finite set of states. A \textit{sample} is a subset of $\Gamma^n$.
A \textit{realization} of the items is an element $a \in \Gamma^{n}$,
representing an assignment of states to items,
where for $i \in N$, $a_i$ represents the state of item $i$.
We also refer to an
element of $\Gamma^n$ as an \textit{assignment}.

We call $b\in\left(\Gamma \cup \left\{*\right\}\right)^{n}$ a
\textit{partial} realization.  Partial realization $b$ represents the
subset of items $I=\{i \mid b_i \neq *\}$ where each item $i \in I$ has
state $b_i$.  For $\gamma \in \Gamma$, the quantity
$b_{i\leftarrow \gamma}$ denotes the partial realization that is
identical to $b$ except that $b_{i}=\gamma$. For partial realizations
$b, b'\in \left(\Gamma\cup \{*\}\right)^{n}$, $b'$ is an
\textit{extension} of $b$, written $b'\succeq b$, if $b'_{i}=b_{i}$ for
all $b_{i}\neq *$. 
We use $b' \succ b$ to denote that $b' \succeq b$ and $b' \neq b$.
% (that is, $b'_{i}=b_{i}$ for all $b_{i}\neq*$ \textbf{and}
% $b' \neq b$).

% Given realization $a \in \Gamma^n$ and $U \subseteq N$, the {\em
%  restriction of $a$ to $U$} is the partial realization $b$ such that
% $b_i = a_i$ for $i \in U$ and $b_i = *$ otherwise.

Let
$g\colon\left(\Gamma \cup \{*\}\right)^{n}\rightarrow \mathbb{Z}_{\geq
  0}$
be a utility function.  
% We assume that $g$ is given by an oracle that
% takes constant time to answer a query.
Utility function
$g\colon(\Gamma \cup \{*\})^{n}\rightarrow \mathbb{Z}_{\geq 0}$ has
{\em goal value} $Q$ if $g(a) = Q$ for all realizations
$a \in \Gamma^n$.

We define $\Delta{g}(b,i,\gamma) := g(b_{i\leftarrow\gamma})-g(b)$.  

A standard utility function is a set function
$f:2^N \rightarrow \mathbb{R}_{\geq 0}$.  It is monotone if for all
$S \subset S' \subseteq N$, $f(S) \leq f(S')$.  It is submodular if in
addition, for $i \in N-S$,
$f(S \cup \{i\}) - f(S) \geq f(S' \cup \{i\}) - f(S')$.  We
extend the definitions of
monotonicity and submodularity to (state-dependent) utility function
$g\colon\left(\Gamma \cup \{*\}\right)^{n}\rightarrow \mathbb{Z}_{\geq 0}$
as follows:

\begin{itemize}
\item $g$ is \textit{monotone} if for
  $b\in\left(\Gamma \cup \{*\}\right)^{n}$, $i\in N$ such that
  $b_{i}=*$, and $\gamma\in\Gamma$, we have
  $g(b) \leq g(b_{i\leftarrow\gamma})$
\item $g$ is \textit{submodular} if for all
  $b,b'\in \left(\Gamma\cup\{*\}\right)^{n}$ such that $b'\succ b$,
  $i\in N$ such that $b_{i}=b'_{i}=*$, and $\gamma\in\Gamma$, we have
  $\Delta{g}(b,i,\gamma) \geq \Delta{g}(b'i,\gamma)$.
\end{itemize}

% We will restrict our attention to integer-valued utility functions 
% $g\colon(\Gamma \cup \{*\})^{n}\rightarrow \mathbb{Z}_{\geq 0}$.

Let $\mathcal{D}$ be a probability distribution on $\Gamma^n$.
% We now define the adaptive submodularity property for $g$ and
% ${\cal D}$.
Let $X$ be a random variable drawn from $\mathcal{D}$.  For
$a \in \Gamma^n$ and $b \in (\Gamma \cup \{*\})^n$, we define
$\Pr[a \mid b] := \Pr[X=a \mid a \succeq b]$.  For $i$ such
that $b_{i}=*$, we define
$\mathbb{E}[\Delta g(b,i,\gamma)] := \sum_{a\in\Gamma^{n}:a\succeq
  b}\Delta{g}(b,i,a_{i})\Pr[a \mid b]$.

\begin{itemize}
\item $g$ is \textit{adaptive submodular with respect to $\mathcal{D}$}
  if for all $b'$, $b$ such that $b'\succ b$, $i\in N$ such that
  $b_{i}=b'_{i}=*$, and $\gamma \in \Gamma$, we have
  $\mathbb{E}[\Delta{g}(b,i,\gamma)]\geq\mathbb{E}[\Delta{g}(b',i,\gamma)]$.
\end{itemize}

Intuitively, we can view $b$ as partial information about states of items $i$ in a random
realization $a \in \Gamma^n$, with $b_i=*$ meaning the state of item $i$ is unknown. 
Then $g$ measures the utility of that
information, and $\mathbb{E}[\Delta_g(b,i,\gamma)]$ is the expected
increase in utility that would result from discovering the state of
$i$.

For $g\colon(\Gamma \cup \{*\})^{n}\rightarrow \mathbb{Z}_{\geq 0}$ with
goal value $Q$, 
and $b \in (\Gamma \cup \{*\})^n$ and $i \in N$, where $b_i = *$,
let
$\gamma_{b,i}$ be the state $\gamma \in \Gamma$ such that
$\Delta g(b,i,\gamma)$ is minimized (if more than one
minimizing state exists, choose one arbitrarily).  Thus $\gamma_{b,i}$
is the state of item $i$ that would produce the smallest increase in
utility, and thus is ``worst-case'' in terms of utility gain, if we
start from $b$ and then discover the state of $i$.

For fixed $g\colon(\Gamma \cup \{*\})^{n}\rightarrow \mathbb{Z}_{\geq 0}$ with
goal value $Q$, we define an associated quantity $\rho$, as follows:
\[\rho := \min \frac{\Delta g(b,i,\gamma)}{Q-g(b)}\]
where the minimization is over $b,i,\gamma$,
where $b \in (\Gamma \cup \{*\})^n$ such that $g(b) < Q$,
$i \in N$, $b_i = *$, and $\gamma\in\Gamma - \{\gamma_{b,i}\}$.

Intuitively, right before the state of an item $i$ is discovered,
there is a certain distance from the current utility achieved to the goal utility.
When the state of that item is discovered, the distance to goal is reduced by
some fraction (or possibly by zero).
The size of that fraction can vary depending on the state of the item.
In the definition of $\rho$, we are concerned with the value of that fraction, 
not for the worst-case state in this case
(leading to the smallest fraction), but for the next-to-worst case state.
The parameter $\rho$ is the smallest possible value for this fraction,
starting from any partial realization, and considering any item $i$
whose state is about to be discovered.

An instance of the Scenario SC problem is a tuple $(g,Q,S,w,c)$, where
$g\colon(\Gamma \cup \{*\})^{n}\rightarrow \mathbb{Z}_{\geq 0}$ is an
integer-valued, monotone submodular utility function with goal value
$Q > 0$, $S\subseteq \Gamma^n$, $w:S \rightarrow \mathbb{Z}_{>0}^{n}$
assigns a weight to each realization $a \in S$, and
$c\in \mathbb{R}_{>0}^{n}$ is a {\em cost vector}.  We consider a
setting where we select items without repetition from the set of items
$N$, and the states of the items correspond to an initially unknown
realization $a \in \Gamma^n$.  Each time we select an item, the state
$a_i$ of the item is revealed.  The selection of items can be
adaptive, in that the next item chosen can depend on the states of the
previous items.  We continue to choose items until $g(b) = Q$, where
$b$ is the partial realization representing the states of the chosen
items.

The Scenario SC problem asks for an adaptive order in which to
choose the items (i.e., a {\textit strategy}), until goal value $Q$ is
achieved, such that the expected sum of the costs of the chosen items
is minimized.  The expectation is with respect to the distribution on
$\Gamma^n$ that is proportional to the weights on the assignments in
the sample: $\Pr[a] = 0$ if $a \not\in S$, and
$\Pr[a] = \frac{w(a)}{W}$ otherwise, where $W = \sum_{a \in S} w(a)$.
We call this the {\em sample distribution} defined by $S$ and $w$ and
denote it by $\mathcal{D}_{S,w}$.

The strategy corresponds to a decision tree.  
The internal nodes of the tree are labeled with items $i \in N$, and
each such node has one child for each state $\gamma \in \Gamma$.
Each root-leaf path in the tree is associated with a partial realization
$b$ such that for each consecutive pairs of nodes $v$ and $v'$ on the
path, if $i$ is the label of $v$, and $v'$ is the \textit{$\gamma$-child} of $v$,
then $b_i = \gamma$.  If $i$ does not label any node in the path,
then $b_i = *$.
The tree may be output
in an implicit form (for example, in terms of a greedy rule),
specifyng how to determine the next item to choose, given the previous
items chosen and their states.  Although realizations $a \not\in S$ do
not contribute to the expected cost of the strategy, we require the
strategy to achieve goal value $Q$ on \textit{all} realizations
$a \in \Gamma^n$.

We will make frequent use of a construction that we call the \textit{standard OR construction} (cf.~\cite{guilloryBilmes11,deshpandeHellersteinKletenik}).  
It is a method for combining two
monotone submodular utility functions $g_1$ and $g_2$ defined on
$(\Gamma \cup \{*\})^n$, and values $Q_1$ and $Q_2$, into a
new monotone submodular utility function $g$.
For
$b \in (\Gamma \cup \{*\})^n$,
\[g(b) = Q_1Q_2 - (Q_1 - g_1(b))(Q_2 - g_2(b))\]
Suppose that on any $a\in\Gamma^{n}$, $g_{1}(a)=Q_{1}$ or
$g_{2}(a)=Q_{2}$. Then, $g(a)=Q_{1}Q_{2}$ for all $a\in\Gamma^{n}$.

\section{Mixed Greedy}
\label{sec:MG}

The Mixed Greedy algorithm is a generalization of the approximation algorithm developed by
Cicalese et al.\ for the Equivalence Class Determination problem.  
That algorithm effectively solves
the Scenario Submodular Cover problem for a particular
``Pairs'' utility
function associated with Equivalence Class Determination. 
In contrast, Mixed Greedy
can be used on any monotone, submodular utility function $g$.

% a tighter approximation bound is more important than simplicity or
% efficiency, we provide a generalization for the CLS algorithm. Our
% algorithm is very similar to the standard

% We present the pseudocode for Mixed Greedy below, expressed as a recursive function, together with
% a high-level description of the algorithm.
Following Cicalese et al., we present Mixed Greedy as outputting a
decision tree. If the strategy is only to be used on one realization, it is not
necessary to build the entire tree. While Mixed Greedy is very similar to
the algorithm of Cicalese et~al, we describe it fully here so
that our presentation is self-contained.

\subsection{Algorithm}

The Mixed Greedy algorithm builds a decision tree for 
Scenario SC instance
$(g,Q,S,w,c)$.
The tree is built top-down.  It has approximately
optimal expected cost, with respect to the sample distribution 
$\mathcal{D}_{S,w}$
defined by $S$ and $w$.
Each internal node of the constructed tree has $|\Gamma|$
children, one corresponding to each state $\gamma \in \Gamma$.  We
refer to the child corresponding to $\gamma$ as the $\gamma$-child.

The Mixed Greedy algorithm works by calling the recursive function
\texttt{MixedGreedy}, whose pseudocode we present in
Algorithm~\ref{alg:main}.
In the initial call to \texttt{MixedGreedy}, $b$ is set to be equal to
$(*,\dots,*)$.  Only the value of $b$ changes between the recursive
calls; the other values remain fixed.  Each call to \texttt{MixedGreedy}
constructs a subtree of the full tree for $g$, rooted at a node $v$ of
that tree.  In the recursive call that builds the subtree rooted at
$v$, $b$ is the partial realization corresponding to the path from the
root to $v$ in the full tree: $b_i = \gamma$ if the path includes a
node labeled $i$ and its $\gamma$-child, and $b_i = *$ otherwise.

\begin{algorithm}
%%\captionsetup[algorithm]{name=MixedGreedy}
  \small \textbf{Procedure} \texttt{MixedGreedy}($g,Q,S,w,c,b$)
  \begin{algorithmic}[1]
    \STATE { \textbf{If} $g(b)=Q$ \textbf{then return} a single (unlabeled)
      leaf $l$} \label{b1start} \STATE { Let $T$ be an empty tree }
    \STATE { $N' \leftarrow \{i:b_i = *\}$ } 
    \STATE { For $i\in N'$, 
      $\sigma_{i} \leftarrow \argmin\limits_{\gamma\in\Gamma} \Delta g(b,i,\gamma)$ }\label{line:sigmadef}
    \label{step:setsigma}
    %%%\STATE {$S' \leftarrow \{a \in S|a \succeq b\}$}
      \STATE {Define 
      $g'\colon 2^{N'}\rightarrow \mathbb{Z}_{\geq 0}$ such that
      for all $U \subseteq N'$, $g'(U) = g(b_U) - g(b)$, where $b_U$
      is the extension of $b$ produced by setting $b_i = \sigma_{i}$
      for all $i \in U$.}  
    \STATE{ $B \leftarrow$ \texttt{FindBudget}$(N', g', c)$,
      $\; spent \leftarrow 0, \; spent_2 \leftarrow 0, \; k\leftarrow 1$} 
      \label{line:budget} \label{b2start}
    \STATE {$I \leftarrow \{i\in N'| c_i \leq B\}$}
    \STATE {For all $R \subseteq I$, define $D_R := \{a \in S|a \succeq b$ and $a_i \neq \sigma_i$ for some $i \in R\}$}
    \STATE {Define $h:2^I \rightarrow \mathbb{Z}_{\geq 0}$ such that for all $R \subseteq I$, $h(R) = \sum_{a \in D_R} w(a)$}
    \STATE {$R \leftarrow \emptyset$}
    \REPEAT \label{b3-start}
    \STATE {Let $i$ be an item which maximizes
      $\frac{h(R \cup \{i\}) - h(R)}{c_i}$ among all items $i \in I$} 
    \label{step:greedychoice1}
    \STATE {Let $t_{k}$ be a new node labeled with item $i$ }
    \STATE {\textbf{If} $k=1$ \textbf{then} make $t_1$ the root of $T$ }
    \STATE {\textbf{else} make $t_k$ the $\sigma_{j}$-child of $t_{k-1}$}
    \STATE {$j\leftarrow i$}
    \FOR { every $\gamma \in \Gamma$ such that
      $ \gamma \neq \sigma_{i}$} \label{line:main-Ui1}
    \STATE{ $T^{\gamma} \leftarrow $ \texttt{MixedGreedy}($g,Q,S,w,c,b_{i \leftarrow \gamma}$) }
    \label{step:rec1}
    \STATE { Attach $T^{\gamma}$ to $T$ by making the root of
      $T^{\gamma}$ the $\gamma$-child of
      $t_k$} \label{line:main-Call-1}
    \ENDFOR 
       \STATE
    {$b_{i}\leftarrow \sigma_{i}, \; R \leftarrow R \cup \{i\}, \; 
      I \leftarrow I - \{i\}, \; spent \leftarrow spent + c_i ,\; 
     k\leftarrow k+1$}
    % \ENDWHILE
    \UNTIL{ $spent \geq B$} \label{b3-end}

    \REPEAT \label{b4-start}
    \STATE {Let $i$ be an item which maximizes
      $\frac{\Delta g(b,i,\sigma_i)}{c_{i}}$ among all items $i \in I$}\label{line:main-greedy2}
    \label{step:greedychoice2}
    \STATE { Let $t_{k}$ be a node labeled with item $i$ }
    \STATE { Make $t_{k}$ the $\sigma_{j}$-child of $t_{k-1}$}
    \STATE {$j\leftarrow i$}
    \FOR { every $\gamma \in \Gamma$ such that
      $ \gamma \neq \sigma_{i}$} \label{line:main-Ui2}
    \STATE { $T^{\gamma} \leftarrow $ \texttt{MixedGreedy}($g,Q,S,w,c,b_{i \leftarrow \gamma}$) }
    \label{step:rec2}
    \STATE { Attach $T^{\gamma}$ to $T$ by making the root of
      $T^{\gamma}$ the $\gamma$-child of
      $t_k$} \label{line:main-Call-2}
    \ENDFOR
    \STATE
    {$b_{i} \leftarrow \sigma_{i}, \; I \leftarrow I - \{i\}, \; 
      spent_2 \leftarrow spent_2 + c_i ,\; k\leftarrow k+1$}
    \UNTIL{$spent_2 \geq B$ \textbf{or}  $I = \emptyset$}
    \label{b4-end}

    \STATE{$T' \leftarrow $ \texttt{MixedGreedy}($g, Q, S,w,c,b$);
      Attach $T'$ to $T$ by making the root of $T'$ the
      $\sigma_{j}$-child of $t_{k-1}$ } \label{line:main-Call-3}
     \label{step:lastrec}

    \STATE{Return $T$}
  \end{algorithmic}
  \caption{} \label{alg:main} 
\end{algorithm}

\begin{algorithm}
\textbf{Procedure} \texttt{FindBudget}($I,f,c$)
  \begin{algorithmic}[1]
      \STATE {Let
      $\alpha=1-e^{-\chi}\approx 0.35$} \STATE { Do a binary search in
      the interval [0,$\sum_{i\in I} c_i$] to find the smallest $B$
      such that Wolsey's greedy algorithm for maximizing a submodular
      function within a budget of $B$, applied to $f$ and the items in $I$,
      returns a set of items with utility at least
      $\alpha f(I)$}
    \STATE { Return $B$}
  \end{algorithmic}
  \normalsize
\end{algorithm}

% For the call to Mixed Greedy building the subtree rooted at $v$,
% let $N' = \{i \in N|b_i = *\}$. So $N'$ consists of the items
% not already chosen on the root-leaf path to node $v$ in the full tree.
% (Note that the sum of the $p(a)$, for $a \in V(b_v)$, can be less
% than 1.)  We can view this as an instance of the Sample-Based ASC
% problem by scaling the proabilities $p(a)$ so that they sum to 1
% (i.e., replacing $p(a)$ by the conditional probability of $a$ in
% ${\cal D}$, given that it is an extension of $b$), restricting the
% inputs to the items $i$ that have not yet been tested,

The algorithm of Cicalese et al.\ for the Equivalence Class
Determination problem is essentially the same as our Mixed Greedy
algorithm, for $g$ equal to their ``Pairs''
utility function.  (There is one small difference -- in their
algorithm, the first stage ends right before the greedy step in which
the budget $B$ would be exceeded, whereas we allow the budget to be
exceeded in the last step.)  Like their algorithm, our Mixed Greedy
algorithm relies on a greedy algorithm for the Budgeted Submodular
Cover problem due to Wolsey.  We describe Wolsey's algorithm in detail
in Appendix~\ref{sec:wolsey}.

If $g(b) = Q$, then \texttt{MixedGreedy} returns an (unlabeled) single node,
which will be a leaf of the full tree for $g$.  Otherwise, \texttt{MixedGreedy}
constructs a tree $T$.
It does so by computing a special realization called $\sigma$, and then
iteratively using $\sigma$ to construct a
path descending from the root of this subtree, which is called the
\emph{backbone}.  It uses recursive calls to build the subtrees
``hanging'' off the backbone. The backbone has a special property: for each
node $v'$ in the path, the successor node in the path is the
$\sigma_i$ child of $v'$, where $i$ is the item labeling node $v'$.

The construction of the backbone is done as follows.  Using subroutine
\texttt{FindBudget}, \texttt{MixedGreedy} first computes a lower bound $B$ on the
minimum additional cost required in order to achieve a portion
$\alpha$ of the goal value $Q$, assuming we start with partial
realization $b$ (Step~\ref{line:budget}).  This computation is done
using the Greedy algorithm of~\cite{wolsey82} described in
Section~\ref{sec:wolsey} in the Appendix.
% The lower bound $B$ also has a special property that we
% describe below.

After calculating $B$, \texttt{MixedGreedy} constructs the backbone in two
stages, using a different greedy criterion in each to determine which
item $i$ to place in the current node.  In the first stage,
corresponding to the first repeat loop of the pseudocode, the goal is
to remove weight (probability mass) from the backbone, as cheaply and as soon
as possible.  That is, consider a realization $a \in \Gamma^n$ to be
removed from the backbone (or ``covered'') if $i$ labels a node in the
spine and $a_i \neq \sigma_{i}$; removing $a$ from the backbone
results in the loss of weight $w(a)$ from the backbone.  The
greedy choice used in the first stage in Step~\ref{step:greedychoice1}
follows the standard rule of maximizing {\em bang-for-the-buck}; the
algorithm chooses $i$ such that the amount of probability mass removed
from the backbone, divided by the cost $c_i$, is maximized.  However,
in making this greedy choice, it only considers items that have cost
at most $B$.  The first stage ends as soon as the total cost of the
items in the chosen sequence is at least $B$.  
For each item $i$ chosen during the stage,
$b_i$ is set to $\sigma_i$.

% That is, let $h_b$ denote a standard utility function defined on the
% subsets $N'$ of the untested items of $b$, such that $h_b(N')$ is
% the sum of $p(a)$, for all $a \in V(b)$ such that $a_i \neq d_{b,i}$
% for some $i \in N'$.

In the second stage, corresponding to the second repeat loop, the goal
is to increase utility as measured by $g$, under the assumption that
we already have $b$, and that the state of each remaining item $i$ is
$\sigma_{i}$.  The algorithm again uses the bang-for-the-buck rule,
choosing the $i$ that maximizes the increase in utility, divided by
the cost $c_i$ (Step~\ref{step:greedychoice2}).  In making this greedy
choice, it again considers only items that have cost at most $B$.  The
stage ends as soon as the total cost of the items in the chosen
sequence is at least $B$.
% We make a small change here as compared to the algorithm of Cicalese
% et al. to slightly simplify the analysis; we assume that the greedy
% sequence is chosen with respect to $b$, ignoring the contribution of
% the items chosen in the first stage.  (If an item in this second
% sequence was already chosen in the first stage, we do not need to
% choose it again.)
For each item $i$ chosen during the stage,
$b_i$ is set to $\sigma_i$.

In Section~\ref{sec:defs}, we defined the value $\rho$. The way the
value $B$ is chosen guarantees that the updates to $b$ during the two
greedy stages cause the value of $Q-g(b)$ to shrink by at least a
fraction $\rho$ before each recursive call.  In Appendix~\ref{app:mixedgreedy}, we prove
this fact and use it to prove the following theorem.

\begin{theorem}
  \label{thm:mixedgreedy1}
  Mixed Greedy is an approximation algorithm for the Scenario
  Adaptive Submodular Cover problem that achieves an approximation
  factor of $O(\frac{1}{\rho}\log Q)$.
\end{theorem}

\section{Scenario Mixed Greedy}
\label{sec:SMG}
We now present a variant of Mixed Greedy that eliminates the dependence on
$\rho$ in the approximation bound in favor of a dependence on $m$, the size of
the sample.  We call this variant Scenario Mixed Greedy.  

Scenario Mixed Greedy works by first modifying $g$ to produce a new
utility function $g_S$, and then running Mixed Greedy with $g_S$,
rather than $g$.  Utility function $g_S$ is produced by combining $g$
with another utility function $h_S$, using the standard OR construction
described at the end of Section~\ref{sec:defs}.
Here $h_S \colon\left(\Gamma \cup \{*\}\right)^{n}\rightarrow
\mathbb{Z}_{\geq 0}$,
where $h_S(b) = m - |\{a \in S:a \succeq b\}|$ and $m = |S|$.
Thus $h_S(b)$ is the total number of assignments that
have been eliminated from $S$ because they are incompatible with the
partial state information in $b$.  Utility $m$ for $h_S$ is achieved
when all assignments in $S$ have been eliminated.  Clearly,
$h_S$ is monotone and submodular.  

When the OR construction is applied to combine $g$ and $h_S$, the
resulting utility function $g_S$ reaches its goal value $Qm$ when all
possible realizations of the sample have been eliminated or
when goal utility is achieved for $g$.

In an on-line setting, Scenario Mixed Greedy uses the
following procedure to determine the adaptive sequence of items to choose
on an initially unknown realization $a$.

\bigskip
\noindent\textbf{Scenario Mixed Greedy:}
\begin{enumerate}\cramped
\item Construct utility function $g_S$ by applying the standard OR construction
to $g$ and utility function $h_S$.
\item Adaptively choose a sequence of items by running
  Mixed Greedy for utility function $g_S$ with goal value
  $Qm$, with respect to the sample distribution $\mathcal{D}_{S,w}$. 
\item After goal value $Qm$ is achieved, 
if the final partial realization $b$ computed by Mixed Greedy
  does not satisfy $g(b) = Q$, then choose the
  remaining items in $N$ in a fixed but arbitrary order until $g(b) = Q$.
\end{enumerate}

The third step in the procedure is present because 
goal utility $Q$ must be reached for $g$ even on realizations $a$ that are not in $S$.

% [LH:\textit{LH: Say something somewhere about polynomial time?}]

\begin{theorem}
  \label{thm:mixedgreedy2}
  Scenario Mixed Greedy is an approximation algorithm for the
  Scenario Submodular Cover problem that achieves an
  approximation factor of $O(\log Qm)$, where $m$ is the size of sample $S$.
\end{theorem}

\begin{proof}
  Scenario Mixed Greedy achieves utility
  value $Q$ for $g$ when run on any realization $a \in \Gamma^n$, because
  the $b$ computed by Mixed Greedy is such that $a \succeq b$, and
  the third step ensures that $Q$ is reached.

  Let $c(g)$ and $c(g_S)$ denote the expected cost of the optimal
  strategies for the Scenario SC problems on $g$ and $g_S$
  respectively, with respect to the sample distribution
  $\mathcal{D}_{S,w}$.  Let $\tau$ be an optimal strategy for $g$
  achieving expected cost $c(g)$.  It is also a valid strategy for the
  problem on $g_S$, since it achieves goal utility $Q$ for $g$ on all
  realizations, and hence achieves goal utility $Qm$ for $g_S$ on all
  realizations.  Thus $c(g_S) \leq c(g)$.

  The two functions, $g$ and $h_S$, are monotone and submodular. Since the function $g_S$ is
  produced from them using the standard OR construction, $g_S$ is also
  monotone and submodular.  Let $\rho_S$ be the value of parameter $\rho$ for the function $g_S$.
  By the bound in
  Theorem~\ref{thm:mixedgreedy1}, running Mixed Greedy on $g_{S}$, for
  the sample distribution $\mathcal{D}_{S,w}$, has expected cost that is at most
  a $O(\frac{1}{\rho_S}\log Qm)$ factor more than $c(g_S)$.
  Its expected cost is
  thus also within an $O(\frac{1}{\rho_S}\log Qm)$ factor of $c(g)$.  Making
  additional choices on realizations not in $S$, as done in the last
  step of Scenario Mixed Greedy, does not affect the expected
  cost, since these realizations have zero probability.

  Generalizing an argument from~\cite{cicaleseLaberSaettler}, we now
  prove that $\rho_S$ is lower bounded by a constant fraction.
  Consider any $b \in (\Gamma \cup \{*\})^n$ and $i \in N$ such that
  $b_i = *$, and any $\gamma \in \Gamma$ where
  $\gamma \neq \gamma_{b,i}$.  Let
  $C_b = |S| - h_S(b) = |\{a \in S \mid a \succeq b\}|$.  Since the
  sets $\{a \in S \mid a \succeq b$ and $a_i = \gamma\}$ and
  $\{a \in S \mid a \succeq b$ and $a_i = \gamma_{b,i}\}$ are
  disjoint, it is not possible for both of them to have size greater
  than $\frac{C_b}{2}$.  It follows that
  $\Delta{h_S}(b,i,\gamma) \geq \frac{C_b}{2}$ or
  $\Delta{h_S}(b,i,\gamma_{b,i}) \geq \frac{C_b}{2}$ or both.  By the
  construction of $g_S$, it immediately follows that
  $\Delta{g_S}(b,i,\gamma) \geq \frac{(Q-g(b))C_b}{2}$ or
  $\Delta{g_S}(b,i,\gamma_{b,i}) \geq \frac{(Q-g(b))C_b}{2}$ or both.
  Since $\gamma_{b,i}$ is the ``worst-case'' setting for $b_i$ with
  respect to $g_S$, it follows that
  $\Delta{g_S}(b,i,\gamma) \geq \Delta{g_{S}}(b,i,\gamma_{b,i})$, and
  so in all cases
  $\Delta{g_{S}}(b,i,\gamma) \geq \frac{(Q-g(b))C_{b}}{2}$.
  Also, $(Q - g(b))C_b = Qm - g_S(b)$. Therefore,
  $\rho_S \geq \frac{1}{2}$. The theorem follows from the bound given
  in Theorem~\ref{thm:mixedgreedy1}.
\end{proof}

\section{Scenario Adaptive Greedy}
\label{sec:SAG}

Scenario Adaptive Greedy works by first constructing a utility
function $g_W$, produced by applying the standard OR construction to
$g$ and utility function $h_W$.  Here
$h_W \colon\left(\Gamma \cup \{*\}\right)^{n}\rightarrow
\mathbb{Z}_{\geq 0}$,
where $h_W(b) = W - \sum_{a \in S:a \succeq b} w(a)$.  Intuitively,
$h_W(b)$ is the total weight of assignments that have been eliminated
from $S$ because they are incompatible with the partial state
information in $b$.  Utility $W$ is achieved for $h_W$ when all
assignments in $S$ have been eliminated.  It is obvious that $h_W$ is
monotone and submodular.  The function $g_W$ reaches its goal value
$QW$ when all possible realizations of the sample have been eliminated
or when goal utility is achieved for $g$.  Once $g_W$ is constructed,
Scenario Adaptive Greedy runs Adaptive Greedy on $g_W$.

In an on-line setting, Scenario Adaptive Greedy uses the following
procedure to determine the adaptive sequence of items to choose on an
initially unknown realization $a$.

\bigskip
\noindent\textbf{Scenario Adaptive Greedy:}
\begin{enumerate}\cramped
\item Construct modified utility function $g_W$ by applying the
  standard OR construction to $g$ and utility function $h_W$.
\item Run Adaptive Greedy for utility function $g_W$ with goal value
  $QW$, with respect to sample distribution $\mathcal{D}_{S,w}$,
  to determine the choices to make on $a$.
\item After goal value $QW$ is achieved, if the partial realization
  $b$ representing the states of the chosen items of $a$ does not
  satisfy $g(b) = Q$, then choose the remaining items in $N$ in
  arbitrary order until $g(b) = Q$.
\end{enumerate}

In Appendix~\ref{app:adaptivesubmod}, we prove the following lemma.
% [\textit{LH: move it here?}]

\begin{lemma}
  \label{lem:gW}
  Utility function $g_W$ is adaptive submodular with respect to sample
  distribution $\mathcal{D}_{S,w}$.
\end{lemma}

The consequence of \lemmaref{lem:gW} is that we may now use any
algorithm designed for adaptive submodular utility functions. This
gives us \theoremref{thm:agforasc}.

\begin{theorem}\label{thm:agforasc}
  Scenario Adaptive Greedy is an approximation algorithm for the
  Scenario Adaptive Submodular Cover problem that achieves an
  approximation factor of $O(\log QW)$, where $W$ is the sum of the
  weights on the realizations in $S$.
\end{theorem}

\begin{proof}
  Since $g_W$ is produced by applying the OR construction to $g$ and
  $h_W$, which are both monotone, so is $g_W$.  By Lemma~\ref{lem:gW},
  $g_{W}$ is adaptive submodular with respect to the sample
  distribution.  Thus by the bound of Golovin and Krause on Adaptive
  Greedy, running that algorithm on $g_{W}$ yields an ordering of
  choices with expected cost that is at most a $O(\log QW)$ factor
  more than the optimal expected cost for $g_W$.  By the analogous
  argument as in the proof of Theorem~\ref{thm:mixedgreedy2}, it
  follows that Scenario Adaptive Greedy solves the Scenario Submodular
  Cover problem for $g$, and achieves an approximation factor of
  $O(\log QW)$.
\end{proof}

\acks{L. Hellerstein thanks Andreas Krause for useful discussions at ETH, and especially for
directing our attention to the bound of Streeter and Golovin for min-sum submodular cover.}

\bibliography{coltscenario}

\appendix

\section{Proof of Bound for Mixed Greedy}\label{app:mixedgreedy}

We first discuss the algorithm of Wolsey used in \texttt{FindBudget}.

\subsection{Wolsey's Greedy Algorithm for Budgeted Submodular Cover}
\label{sec:wolsey}
The Budgeted Submodular Cover problem takes as input a finite set
$N$ of items, a positive integer $B > 0$ called the
{\em budget}, a monotone submodular set function
$f:2^N \rightarrow \mathbb{Z}_{\geq 0}$, and a vector
$c$ indexed by the items in $N$, such that $c_i \in \mathbb{R}_{\geq 0}$ for all $i \in N$.
The problem is to find a subset
$R \subseteq N$ such that $\sum_{i \in R} c_i \leq B$, and $f(R)$ is
maximized.

\cite{wolsey82} developed a greedy approximation algorithm for this problem.  We present the pseudocode for this
algorithm here, together with Wolsey's approximation bound.

\begin{algorithm}
%\captionsetup[algorithm]{name=WolseyGreedy}
\label{alg:wolsey}
  \small \textbf{Procedure} \texttt{WolseyGreedy}($N,f,c,B$)
  \begin{algorithmic}[1]
    \STATE { $spent \leftarrow 0$, $R \leftarrow \emptyset$, $k \leftarrow 0$ }
    \REPEAT
    \STATE {$k \leftarrow k+1$}
    \STATE {Let $i_k$ be the $i \in N$ that minimizes $\frac{f(R \cup \{i\}) - f(R)}{c_i}$ among all $i \in N$ with $c_i \leq B$}
    \STATE {$N \leftarrow N - \{i\}$, $spent \leftarrow spent + c_i$, $R \leftarrow R \cup \{i_k\}$}
    \UNTIL {$spent > B$ or $N = \emptyset$}
    \IF {$f(\{i_k\}) \geq f(R - \{i_k\})$} 
    \STATE{return $\{i_k\}$ }
    \ELSE 
    \STATE {return $R - \{i_k\}$}
    \ENDIF
\end{algorithmic}
\end{algorithm}

\begin{lemma}[\cite{wolsey82}]
\label{lem:wolsey}
Let $R^*$ be the optimal solution to the Budgeted Submodular Cover problem on 
instance $(N,f,c,B)$.
Let $R = \{i_1, \ldots, i_k\}$ be the set of items chosen by running Wolsey-Greedy($N,f,c,B$).
Let $e$ be the base of the natural logarithm, and let $\chi$ be the solution to $e^{\chi} = 2-\chi$.
Then $f(R) \geq (1-e^{-\chi})f(R^*)$.
\end{lemma}

\subsection{Analysis of Mixed Greedy}

% analysis we use to bound
% the approximation factor achieved when using the Mixed Greedy algorithm to solve
% the Sample-Based SC problem.
% As discussed earlier, our analysis of Mixed Greedy is mostly very similar to the analysis of 
% Cicalese et al.\ for their Equivalence Class Determination algorithm,
% generalized to apply to an arbitrary  submodular utility function $g$.
% 
% The one major difference appears when we 
% generalize a key lemma 
% from their paper bounding
% a quantity they called the ``sepcost''.
% Our proof uses a totally different approach.
% 
% We present our generalization of this lemma,
% and then give the full analysis of Mixed Greedy.

% Let $\mathcal{I} = (g,Q,S,w,c)$ be an instance of
% the Scenario SC problem.  
% Let $p \colon \Gamma^n \rightarrow [0,1]$ 
% where $p(a)$ is the probability assigned  to $a$ by sample
% distribution $\mathcal{D}_{S,w}$.

% In the recursive calls, $b \neq *^n$ and 

Consider a Scenario SC instance $(g,Q,S,w,c)$, and
a partial realization $b \in (\Gamma \cup \{*\})^n$.
We now consider \texttt{MixedGreedy}($g,Q,S,w,c,b$).  
It
constructs a tree for the Scenario SC instance 
{\em induced by} $b$.  
In this induced instance, the item set is 
$N' = \{i \mid b_i = *\}$.
Without loss of generality, assume that
$N' = \{1, \ldots, n'\}$ for some $n'$.
For $d \in (\Gamma \cup \{*\})^n$ such that $d \succeq b$, define $\nu(d)$ be the restriction of $d$
to the items in $N'$.  
For $d' \in (\Gamma \cup \{*\})^{n'}$,
$\nu^{-1}(d')$ denotes the extension $d \succeq d'$ to all elements in $N$ such that
$d_i = d'_i$ for $i \in N'$ and $d_i = b_i$ otherwise.

The utility function $g':(\Gamma \cup \{*\})^{n'} \rightarrow \mathbb{Z}_{\geq 0}$ for the instance induced by $b$
is a function on partial realizations $d'$ of the items in $N'$.
Specifically,
for $d' \in (\Gamma \cup \{*\})^{n'}$,
$g'(d') = g(\nu^{-1}(d'))$.
The sample
$S'$ in the induced instance consists of the restrictions of the
realizations in $\{a \in S \mid a \succeq b\}$ to the items in $N'$.
That is, $S' = \{\nu(a) \mid a \in S, a \succeq b\}$.
Note that
each realization in $S'$ corresponds to a unique realization in $S$.
The weight function $w'$ for the induced instance 
is such that for all $d' \in S'$,
$w'(d') = w(\nu^{-1}(d'))$.
The goal value for the induced instance is $Q$.

If $g(b) = Q$, then \texttt{MixedGreedy}($g,Q,S,w,c,b$) returns the optimal tree for
the instance induced by $b$, which is a single (unlabeled) leaf with expected cost 0.
Assume $g(b) < Q$.

For any decision tree $\tau$ for the induced instance and any
realization $a$ defined over the item set $N'$ (or over any superset
of $N'$), let $\kappa(\tau,a) = \sum_{i \in M}c_i$, where $M$ is the
set of items labeling the nodes on the root-leaf path followed in
$\tau$ on realization $a$.  That is, $\kappa(\tau,a)$ is the cost
incurred when using tree $\tau$ on realization $a$.

Let $\tau^*$ be a decision tree that is an optimal solution for the induced instance.
Let $C^*=\mathbb{E}[\kappa(\tau^*,a)]$ where $a$ is a random realization drawn
from $D_{S',w'}$.  Thus $C^*$ is the expected cost of an optimal solution to the
induced instance.
Let $\tau^G$ denote the tree
output by running \texttt{MixedGreedy}($g,Q,S,w,c,b$).

Let $\sigma \in \Gamma^{n'}$ be
such that for $i \in N'$, $\sigma_i = \argmin\limits_{\gamma\in\Gamma} g(b_{i\leftarrow \gamma})$.
Thus, $\sigma$ is the realization whose entries are computed in Step~\ref{step:setsigma}
of \texttt{MixedGreedy}.

For each node $v$ in the tree $\tau^G$, let $\tilde{p}(v)$ denote the
probability that node $v$ will be reached when using $\tau^G$ on a
random realization $a$ drawn from $D_{S',w'}$. Let $c_v = c_i$ where $i$
is the item labeling node $v$.  Consider the backbone constructed
during the call to \texttt{MixedGreedy}($g,Q,S,w,c,b$).  The backbone consists
of the nodes created during the two repeat loops in this call, excluding the recursive calls.
Let $Y$ be the set of nodes in the backbone.  Let
$c_Y = \sum_{v \in Y}\tilde{p}(v)c_v$.  Thus $c_Y$ is the contribution
of the nodes in the backbone to the expected cost of tree $\tau^G$.
The following lemma says that this contribution is no more than a
constant times the expected cost of the optimal tree $\tau^*$.

\begin{lemma}
\label{lem:newlem}
$c_Y \leq 24C^{*}$.
\end{lemma}

Lemma~\ref{lem:newlem} is the key technical lemma
in our analysis, and it is the proof of this lemma that constitutes
the major difference between our analysis and the analysis in~\cite{cicaleseLaberSaettler}.
We defer the proof of this lemma to Section~\ref{sec:newlemproof}.
Using this lemma, it is easy to generalize
the rest of the analysis of Cicalese et al.\ to obtain the proof of Theorem~\ref{thm:mixedgreedy1}.
The proofs in the remainder of this section closely follow the proofs in Cicalese et al.
We present them so that this paper will be self-contained.

Let $B$ be the budget that is computed in Line~\ref{line:budget}, with \texttt{FindBudget}, when running
\texttt{MixedGreedy}($g,Q,S,w,c,b$).   
Recall the constant $\alpha$ defined in \texttt{FindBudget}, based on the bound on Wolsey's
Greedy algorithm (Lemma~\ref{lem:wolsey}).

\begin{lemma}
  The condition at the end of the first repeat loop (spent $\geq B$)
  will be satisfied. Also, $\kappa(\tau^*,\sigma) \geq B$.
\end{lemma}

\begin{proof}
% We have $g'(U) \geq Q-g(s)$ since
  Trees $\tau^G$ and $\tau^*$ must achieve utility $Q-g(b)$ on
  realization $\sigma$.  The binary search procedure in \texttt{FindBudget}
  finds the least budget $B$ allowing Wolsey's greedy algorithm to
  achieve a total increase in utility of at least
  $\alpha (Q - g(b))$, on realization $\sigma$.  It follows from the
  bound on Wolsey's greedy algorithm (Lemma~\ref{lem:wolsey}) that on
  realization $\sigma$, an increase of $\alpha(Q - g(b))$ could not be
  achieved with a budget smaller than $B$. Thus,
  $\kappa(\tau^*, \sigma) \geq B$.
\end{proof}

% Here's a fuller proof but I don't think we need it.
% Let $B' = \kappa(\tau^*,\sigma)$.  Thus $B'$ is the 
% sum of the costs of the set of items $U$ labeling the nodes of the
% path traversed when using $\tau^*$ on realization $\sigma$.
% Also, $g'(U) \geq Q$,
% since $\tau^*$ must achieve utility value $Q$ on all realizations.
% Therefore, when Wolsey's algorithm is run on $g'=g$, $N' = n$, and $c$,
% with any budget that is at least $B'$, by the approximation bound for Wolsey's
% algorithm (Lemma~\ref{}) the utility achieved will be
% at least $\alpha Q$.
% Thus the binary search algorithm will return a value of at most $\kappa(\tau^*,\sigma)$.

The next lemma clearly holds because in the two repeat loops, 
we only consider items of cost at most $B$, and we 
continue choosing items of cost at most $B$ until a budget of $B$ is met or exceeded.

\begin{lemma}
$\sum_{v \in Y} c_v\leq 4B$.
\end{lemma}

Let $b^{final}$ denote the final value of $b$ in the last recursive call, in Line~\ref{step:lastrec},
when running \texttt{MixedGreedy}($g,Q,S,w,c,b$).

\begin{lemma}
\label{lem:bfinal}
$g(b^{final}) \geq g(b) + \frac{1}{9}(Q - g(b))$.
\end{lemma}

\begin{proof}
  Recall that $N'=\{1, \ldots, n'\}$. For any $D \subseteq N'$, let
  $\hat{\sigma}^D$ denote the extension of $b$, to
  $(\Gamma \cup \{*\})^{n}$, such that
  $\hat{\sigma}^D_i = \sigma_i$ (as specified in line~\ref{line:sigmadef} of
  \texttt{MixedGreedy}($g,Q,S,w,c,b$)) for $i \in D$, and
  $\hat{\sigma}^{D}_{i}=b_{i}$ otherwise.

It follows from the way that $B$ was computed in \texttt{FindBudget},
and the fact that the value of $g$ is $Q$ on any (full) realization of the items in $N$,
that there is
a subset $L\subseteq N'$ such that $\sum_{i \in L} c_i = B$
and $g(\hat{\sigma}^L) \geq \alpha (Q-g(b)) + g(b)$.

% Assuming that the state of every item $i \in N$ is $\sigma_i$.
Let $Y_1$ and $Y_2$ be the set of items $i$ chosen in the first and
second repeat loops respectively.  Thus
$b^{final} = \hat{\sigma}^{Y_1 \cup Y_2}$.

Let $d_{1} = g(\hat{\sigma}^{Y_{1}}) - g(b)$ represent the utility
gained in the first repeat loop. Let
$d_{2}=g(\hat{\sigma}^{Y_{1}\cup L}) - g(\hat{\sigma}^{Y_{1}})$
represent the additional utility that the items in $L\setminus Y_{1}$
would provide.  Since
$g(\hat{\sigma}^{L}) \geq \alpha (Q-g(b)) + g(b)$ and $g$ is monotone,
$g(\hat{\sigma}^{Y_1 \cup L}) \geq g(\hat{\sigma}^L)$, and thus
$g(\hat{\sigma}^{Y_1 \cup L}) \geq \alpha(Q - g(b))+g(b)$. So
$d_{1}+d_{2} \geq \alpha(Q-g(b))$. At the end of the first repeat loop
the items in $Y_1$ have been chosen.  If we were to add the items in
$L \setminus Y_1$ to those in $Y_1$, it would increase the utility by
$d_{2}\geq\alpha(Q-g(b)) - d_{1}$.  Since the items in the second
repeat loop are chosen greedily with respect to $g$ (and $c$) until
budget $B$ is met or exceeded, or goal value $Q$ is attained, it
follows by the approximation bound on Wolsey's algorithm
(Lemma~\ref{lem:wolsey}) that the amount of additional utility added
during the second repeat loop is at least $\alpha$ times the amount of
additional utility that would be added by instead choosing the items
in $L \setminus Y_1$.  We thus have
$g(\hat{\sigma}^{Y_{1}\cup Y_{2}}) - g(\hat{\sigma}^{Y_{1}}) \geq
\alpha d_{2}$.
Adding $d_1$ to both sides, from the definition of $d_1$ we get
$g(\hat{\sigma}^{Y_{1}\cup Y_{2}}) - g(b) \geq d_{1} + \alpha
d_{2}$.
We know from above that $d_{2} \geq \alpha (Q-g(b))-d_{1}$ so we have
$g(\hat{\sigma}^{Y_{1}\cup Y_{2}})-g(b)\geq d_{1} +
\alpha\left(\alpha\left(Q-g(b)\right)-d_{1}\right) \geq
d_{1}+\alpha^{2}\left(Q-g(b)\right)-\alpha d_{1} \geq
\alpha^{2}\left(Q-g(b)\right)$.
The lemma follows because the constant $\alpha^2$ is greater than
$\frac{1}{9}$.
\end{proof}

We can now give the proof of Theorem~\ref{thm:mixedgreedy1}, stating that the
Mixed Greedy algorithm
achieves an approximation factor of $O(\frac{1}{\rho}\log Q)$.

\bigskip
%\begin{proof}[of Theorem~\ref{thm:mixedgreedy1}]
\begin{proof}\proofof{\theoremref{thm:mixedgreedy1}}
The Mixed Greedy algorithm solves the Scenario SC instance $(g,Q,S,w,c)$
by running recursive function \texttt{MixedGreedy}($g,Q,S,w,c,b$). 
In the initial call, $b$ is set to $*^n$.

% During the running of MixedGreedy($g,Q,p,c,s$), 
% each recursive call on a partial realization $b$
% produces a tree whose nodes are labeled
% with items from the set $N' = \{i:b_i = *\}$.

Let $\tau^G$ denote the tree that is output by
running \texttt{MixedGreedy}($g,Q,S,w,c,b$). 
Let $\tau^*$ denote the optimal tree for the Scenario SC instance
induced by $b$.

The
expected cost of $\tau^G$ can be broken into the part that is due to
costs incurred on items in the backbone in the top-level
call to the \texttt{MixedGreedy} function, and costs incurred in the
subtrees built in the recursive calls to \texttt{MixedGreedy}.  The recursive
calls in Steps~\ref{step:rec1} and~\ref{step:rec2} build subtrees of
$\tau^G$ that are rooted at a $\gamma$-child of a node labeled $i$,
such that $\gamma \neq \sigma_i$.  It follows from the definition of
$\rho$ that the value of the partial realization used in each of these
recursive calls, $b_{i \leftarrow \gamma}$ is such that
$g(b_{i \leftarrow \gamma}) - g(b) \geq \rho (Q-g(b))$, so
$g(b_{i \leftarrow \gamma}) \geq \rho (Q-g(b))+g(b)$,

The remaining recursive call is performed on $b^{final}$, 
and by Lemma~\ref{lem:bfinal}, $g(b^{final}) \geq \frac{1}{9}(Q-g(b))$.

Let $\eta = \min \{\rho,\frac{1}{9}\}$.  Let $b^1, \ldots, b^t$ denote
the partial realizations on which the recursive calls are made,
and for which the value of $g$ on the partial realization is strictly less than $Q$.
These are the recursive calls which result in the construction of non-trivial subtrees,
with non-zero cost.
Note that $b^1, \ldots, b^t$ may include $b^{final}$. 
For all $j \in \{1, \ldots, t\}$,
$g(b^j) \geq \eta(Q - g(b))+ g(b)$, or equivalently
\begin{equation}
\label{eq:Qreduction}
Q -g(b^j) \leq (1-\eta)(Q-g(b))
\end{equation}
For $j \in \{1,\ldots, t\}$, let $\tau^G_j$ denote the tree returned
by the recursive call on $b^j$.

Let $S'$ be the sample for the Scenario SC instance induced by $b$,
so $S' = \{\nu(a) \mid a \in A \}$.
Let $w'$ be the weight function for that induced instance.
Let
$A_j = \{\nu(a) \mid a \in S, a \succeq b^j\}$.  Let $\mu^*_j$ denote an
optimal decision tree for the Scenario SC instance induced by
$b^j$.  
Consider the optimal decision tree $\tau^*$ for the
instance induced by $b$, and use it to form a decision tree $\tau^*_j$ for the
instance induced by $b^j$ as follows: for each item $i$ such that $b_{i}=*$ and $b^j_i \neq *$,
fix $i$ to have state $b^j_i$ in the tree. That is, for any node in the tree labeled $i$, delete all its children
except the one corresponding to state $b^j_i$, and then delete the node, connecting the parent of the node to its 
one remaining child.
Since
$\mu^*_j$ is optimal for the induced problem, $\tau^*_j$ cannot have
lower expected cost for this problem.  It follows that
$\sum_{a \in A_j}w'(a)\kappa(\tau^*_j,a) \geq \sum_{a \in
  A_j}w'(a)\kappa(\mu^*_j,a)$.
Further, since $\kappa(\tau^*,a) \geq \kappa(\tau^*_j,a)$ for any
$a \in A_j$,
\begin{equation}\label{eq:kappaineq}
  \sum_{a \in A_j}w(a)\kappa(\tau^*,a) \geq \sum_{a \in
    A_j}w'(a)\kappa(\mu^*_j,a).
\end{equation}

From the description of \texttt{MixedGreedy}, 
it is easy to verify that the $A_j$ are disjoint subsets of $S'$.
Therefore, 
\[
\sum_{a \in S'} w'(a)\kappa(\tau^*,a) = \sum_{j=1}^t \sum_{a \in A_j} w'(a)\kappa(\tau^*,a)
\]

Let $W = \sum_{a \in S'} w'(a)$.
For $a \in S'$, let $p(a)$ be the probability assigned to $a$ by distribution
$D_{S',w'}$, so
$p(a) = w'(a)/W$.
Let $c_Y$ be the sum of the costs incurred on the backbone of $\tau^G$ as in Lemma~\ref{lem:newlem}.
Taking expectations with respect to $D_{S',w'}$, we have
$\mathbb{E}[\kappa(\tau^G,a)] = c_Y + \sum_{j = 1}^t \sum_{a \succeq b^j} p(a)\kappa(\tau^G_j,a)$.
We can now bound the ratio between 
$G = \mathbb{E}[\kappa(\tau^G,a)]$ and
$C^* = \mathbb{E}[\kappa(\tau^*,a)]$.

\begin{align*}
\frac{G}{C^*} & = \frac{\sum_{a \in S'} w'(a)\kappa(\tau^G,a)}{\sum_{a \in S'} w'(a)\kappa(\tau^*,a)}\\
&= \frac{Wc_Y + \sum_{j = 1}^t \sum_{a \in A_j} w'(a)\kappa(\tau^G_j,a)}{\sum_{a \in S'} w'(a)\kappa(\tau^*,a)}\\
&= \frac{Wc_Y} {\sum_{a \in S'} w'(a)\kappa(\tau^*,a)} + \frac{\sum_{j = 1}^t \sum_{a \in A_j} w'(a)\kappa(\tau^G_j,a)}{\sum_{a \in S'} w'(a)\kappa(\tau^*,a)}\\
&\leq
24 + \frac{\sum_{j = 1}^t \sum_{a \in A_j} w'(a)\kappa(\tau^G_j,a)}{\sum_{a \in S'} w'(a)\kappa(\tau^*,a)} \tag*{by Lemma~\ref{lem:newlem}}\\
&= 24 + \frac{\sum_{j=1}^t\sum_{a \in A_j} w'(a)\kappa(\tau^G_j,a)}{\sum_{j=1}^t \sum_{a \in A_j} w'(a)\kappa(\tau^*,a)} \\
&\leq 24 + \max_j\frac{\sum_{a \in A_j} {w'(a)\kappa(\tau^G_j,a)}}{\sum_{a \in A_j} w'(a)\kappa(\mu^*_j,a)} \\
\end{align*}
In the last line, we substitute $\kappa(\tau^*,a)$ with
$\kappa(\mu^*_{j},a)$ because of \eqref{eq:kappaineq}, and we use the
max because of the fact that
$\frac{\sum x_{i}}{\sum y_{i}} \leq \max\limits_{i}
\frac{x_{i}}{y_{i}}$ for $x_{i}, y_{i}>0$.

As described above, for each $j$, the
recursive call to \texttt{MixedGreedy} on $b=b^j$ constructs a tree $\tau^G_j$
for a Scenario SC instance $I'$ induced by $b^j$, with goal value
$Q - g(b^j)$.  The tree $\mu^*_{j}$ is an optimal tree for instance
$I'$.  It follows that the ratio
$\frac{\sum_{a \in A_{j'}} {w(a)\kappa(\tau^G_j,a)}}{\sum_{a \in A_j}
  w(a)\kappa(\mu^*_j,a)}$
is equal to $\frac{G_j}{C^*_j}$, where $G_j$ and $C^*_j$ are the
values of $C^*$ and $G$ for the induced instance $I'$.  Thus we have
$\frac{G}{C^*} \leq 24 + \max_j \frac{G_j}{C^*_j}$.

We now prove that
$\frac{G}{C^*} \leq 1+ 24\frac{1}{\eta}\ln (Q-g(b))$, when $g(b) < Q$, by induction on
the total number of items $n=|N|$.
The base case $n=1$ clearly holds.
Assume inductively that
$\frac{G}{C^*} \leq 1+ 24\frac{1}{\eta}\ln (Q-g(b))$ when the number of
items is less than $n$, where $Q$ is the goal value.
Then for $n$ items, we have
$\frac{G}{C^*} \leq 24+(1 + 24\frac{1}{\eta}(\ln (Q-g(b^j))))$ for 
the $j$ maximizing $\frac{G_j}{C^*_j}$.
By~\eqref{eq:Qreduction}, $Q - g(b^j) \leq (1-\eta)(Q-g(b))$
so 
\begin{align*}
  \frac{G}{C^*}
  &\leq 24+\left(1+24\frac{1}{\eta}\ln ((1-\eta)(Q-g(b))\right)\\
  &\leq 1+24\left(1+\frac{1}{\eta}\ln ((1-\eta)(Q-g(b))\right)\\
  &= 1+24\left(1+\frac{1}{\eta}\ln (1-\eta)+\frac{1}{\eta} \ln (Q-g(b))\right)\\
  &\leq 1+24\frac{1}{\eta} \ln (Q-g(b))
\end{align*}
where the last inequality holds because $1-\eta \leq e^{-\eta}$ so
$\log(1-\eta) \leq -\eta$ and thus $\frac{1}{\eta}\ln(1-\eta) \leq -1$.

Since $Q \geq Q-g(b)$, the expected cost of the greedy tree $\tau^G$
constructed by the Mixed Greedy algorithm is within an
$O(\frac{1}{\eta} \ln Q)$ factor of the expected cost of the optimal
tree. Also, since $\eta = \min\{\rho, \frac{1}{9}\}$, we know that
$\frac{1}{\eta}$ is either constant or it is equal to
$\frac{1}{\rho}$. We therefore have that the expected cost of
$\tau^{G}$ is within an $O(\frac{1}{\rho}\log Q)$ factor of the
expected cost of the optimal tree.
\end{proof}

\subsection{Proof of Lemma~\ref{lem:newlem}}
\label{sec:newlemproof}
We now present our proof bounding the expected cost incurred on the
backbone of the greedy tree.  Our proof relies heavily on the work
of~\cite{streeterGolovin} on the Min-Sum Submodular Cover problem.
We use some of their terminology and definitions in our proof.

\subsubsection{Definitions}
We begin by defining a discrete version of the Min-Sum Submodular
Cover problem.  Let $N = \{1, \ldots, n\}$ be a set of items, and let
$c \in \mathbb{Z}^{n}_{\geq 0}$ be a non-negative integer vector of
``times'' associated with those items.  Let
$f:2^N \rightarrow \mathbb{Z}_{\geq 0}$ be a monotone, submodular
utility function and let $Q = f(N)$.  We define a \emph{schedule} to be
a finite sequence $S=\langle (i_1,\tau_1), \ldots, (i_m,{\tau_m}) \rangle$ of
pairs in $N \times \mathbb{R}_{\geq 0}$ and refer to $\tau_{j}$ as the
\emph{time to process item $i_j$}.

For a schedule $S$, we define $\ell(S) = \sum_{j \geq 1} \tau_j$ to be
the sum of the times spent on all items in $S$.  Given a schedule
$S = \langle (v_1, \tau_1), (v_2, \tau_2), \dots\rangle$, we define
$S_{\langle t\rangle}$ to be the schedule such that for 
$t \leq \ell(S)$,
\[
  S_{\langle t\rangle} = \langle (v_1, \tau_1), (v_2, \tau_2), \dots,
  (v_k, \tau_k), (v_{k+1}, t-\ts\sum_{i=1}^k \tau_i)\rangle
\]
where $k = \max\{j : \sum_{i=1}^j \tau_i < t\}$. For $t>\ell(S)$, we
let $S_{\langle t\rangle}=S$.
We refer to $S_{\langle t\rangle}$ as {\em $S$ truncated at time $t$}.

Let $f^c$ denote the function defined on schedules $S$ such that
$f^c(S) = \frac{1}{f(N)}f(\{i\mid (i,c_{i}) \in S\})$.
Thus, the only pairs $(i,\tau)$ in the schedule that contribute to the value of
$f^c$ are  those for which $\tau = c_i$.  Where $c$ is
understood, we will omit the superscript and use $f$ to denote
both the original utility function on $2^N$, and the function $f^c$
which is defined on schedules.

% We will be working exclusively with schedules 
% $S=\langle (i_1,\tau_1), \ldots, (i_m,{\tau_m}) \rangle$ 
% where $\tau_i = c_i$ for $i < m$, and $\tau_m \leq c_i$.

We define the {\em cost of schedule $S$}, with respect to $f$ and $c$, to be
\begin{equation}\label{eq:ubschedcost}
  \cost(f^{c}, S) = \int_{t=0}^{\ell(S)} 1 - f^{c}(S_{\langle t\rangle} )dt
\end{equation}

We define the {\em Discrete Min-Sum Submodular Cover} Problem on $f$ 
and $c$ to be the problem of finding a
schedule $S$ that achieves $f^c(S)=1$ with minimum cost.

Streeter and~Golovin presented a greedy algorithm for the general
Min-Sum Submodular Cover problem.
In Discrete Min-Sum Submodular Cover, a pair $(i,\tau)$
can only contribute to the utility of a schedule if $\tau = c_i$.
The general problem studied by Streeter and Golovin does not
have this restriction.

% We define a related problem, {\em Bounded Discrete Min-Sum Submodular Cover}, 
% in which we are also given a parameter, $B$,
% representing a bound on the total time of a schedule. 
% In this problem we need to find a schedule $S$ minimizing
% 
% \begin{equation}\label{eq:bschedcost}
% cost(f^c, S_{\langle B \rangle}) = \int_{t=0}^{B} 1 - f^c(S_{\langle t\rangle})dt
% \end{equation}

\subsubsection{Standard Greedy Algorithm for Discrete Min-Sum Submodular Cover}

The algorithm of
Streeter and Golovin for the general Min-Sum Submodular Cover problem uses
a standard greedy approach.
It adds pairs $(i,\tau)$ iteratively to the end
of an initially empty schedule, using the greedy rule of choosing the
pair that will result in the largest increase in utility 
per unit time.  
We call this algorithm {\em Standard Greedy}.

We restrict our attention to the Discrete Min-Sum Submodular Cover problem.
Applied to this problem, Standard Greedy
uses the greedy rule of choosing the pair
$(i,c_i)$ that will result in the largest increase in utility as measured
by $f^c$, per unit time.  The algorithm ends when the constructed schedule $S$
satisfies $f^{c}(S) = 1$.

More formally, Standard Greedy uses the greedy rule below to
construct a greedy schedule $G = \langle (g_1,\tau_1), (g_2,\tau_2), \dots\rangle$,
where each $g_j = i$ for some $i \in N$, and $\tau_i = c_i$.
Since each $\tau_i$ is determined by $g_i$, we drop
the $\tau_i$ from the description of the schedule, and consider
$G$ to be simply a list of actions $g = \langle g_1, g_2, \ldots,\rangle$. 

We
define $G_j = \langle g_1, g_2, \dots g_{j-1}\rangle$, where
$G_1 = \langle\,\rangle$. The action $g_j$ chosen
using the greedy rule is as follows (using $\oplus$ to represent the
concatenation of two schedules):
\begin{equation}\label{eq:greedy} g_j =
  \argmax\limits_{(i,c_i)\mid i \in N}\left\{\frac{f(G_j\oplus
      \langle (i,c_i)\rangle ) - f(G_j)}{c_i}\right\} \end{equation}

% When, for a greedy schedule $G$, $\ell(G) \leq B$, the schedule fits
% the original problem description by Streeter
% and~Golovin~\cite{streeterGolovin}, and one can use their original
% analysis. So we will restrict our analysis to those schedules $G$ for
% which $\ell(G) > B$. 

The following theorem of Streeter and Golovin shows that the schedule constructed by Standard Greedy
has a cost that is within a factor of 4 of the cost achieved by any schedule (including
the optimal schedule).

\begin{theorem}[\cite{streeterGolovin}]
\label{thm:streeter}
Let $I$ be an instance of the Discrete Min-Sum Submodular Cover problem with
time vector $c$, monotone submodular utility function $f$, and item set $N$.
Let $\mathcal{S}$ denote the set of all schedules $S$ for item set $N$ and
cost vector $c$ that satisfy $f^{c}(S) = 1$.
Let $G$ be the schedule constructed by running Standard Greedy algorithm on instance $I$.
Then for all $S \in \mathcal{S}$, $\cost(f^c,G) \leq 4\cost(f^c,S)$.
\end{theorem}

\subsubsection{Bound on Cost of MixedGreedy} 
We now return to our analysis of
\texttt{MixedGreedy}($g,Q,S,w,c,b$). 
As part of our analysis, we will prove a result similar to 
\theoremref{thm:streeter}. 

Without loss of generality, assume
that $b = *^n$.

Recall that $c_Y = \sum_{v \in Y}\tilde{p}(v)c_v$, where $Y$ is the
set of nodes in the backbone, $\tilde{p}(v)$ is the probability that a
random realization will reach node $v$, and $c_v$ is the cost of the
item labeling node $v$.  Let
$S^Y = \langle (i_1,c_{i_1}), \ldots, (i_{k-1},c_{i_{k-1}}) \rangle$
be the schedule such that $i_1, \ldots, i_{k-1}$ is the sequence of
items labeling the nodes in the backbone, from the top of the backbone
and moving downwards.

Define a utility function $h_p:2^N \rightarrow \mathbb{R}_{\geq 0}$
such that for $R \in 2^N$,
$h_p(R) = 1 - \sum_{a \succeq \sigma^R} p(a)$, where $\sigma^R$ is the
realization in $\Gamma^n$ such that $\sigma^R_i = \sigma_i$ for
$i \in R$, and $\sigma^R_i = *$ otherwise. The function $h_{p}$ is
clearly monotone and submodular. Additionally, we can see that
$\sum_{v \in Y}\tilde{p}(v)c_v$ is the cost of schedule $S^Y$ with
respect to utility function utility function $h_p$.

Recall that $\tau^*$ denotes the optimal strategy solving the Scenario
Submodular Cover instance on $g$ and $c$.  Consider the sequence
$j_1, \ldots, j_t$ of items chosen by $\tau^*$ on realization
$\sigma$.  Let
$S^* = \langle (j_1, c_{j_1}), \ldots, (j_t, c_{j_t}) \rangle$. 
The schedule $S^{Y}$ created by \texttt{MixedGreedy} 
is constructed greedily, using the same type of greedy rule as in
\eqref{eq:greedy}.  
However, $S^{Y}$ is constructed
in two stages: 
the first stage greedily chooses
with respect to $h_{p}$, and the second chooses greedily with
respect to an entirely different utility function. We therefore cannot
directly apply \theoremref{thm:streeter} to bound the cost of schedule
$S^{Y}$. We deal with this by using an approach analogous to one used
by \cite{cicaleseLaberSaettler}~(in the analysis of their Equivalence
Class Determination algorithm) that allows us to concentrate only on
the cost of the portion of the schedule constructed during the first
stage.

To do this, we note that schedule $S^Y$ can be expressed as the concatenation of two schedules,
$S^1$ and $S^2$, where $S^1$ contains the $i_j$ chosen during the first
repeat loop, with their costs, and $S^2$ contains the $i_j$ chosen during
the second, also with their costs.
Recall that $\sum_{v \in Y}\tilde{p}(v)c_v$ is the cost of schedule $S^Y$
with respect to $h_p$. 
We can express this cost as follows:

\[ \sum_{v \in Y}\tilde{p}(v)c_v = \int_{t=1}^{\ell(S^1)} 1 - h_p(S^1_{\langle t\rangle} )dt + \int_{t=0}^{\ell(S^2)} 1 - h_p(S^1 \oplus S^2_{\langle t\rangle})dt \]

Note that $\ell(S^2) \leq 2B$, since we have assumed that each
$c_i \leq B$, and the second repeat loop of \texttt{MixedGreedy} ends as soon
as the last item added causes the length of $S^2$ to exceed $B$.
Since $h_p$ is monotone, the value of the second integral is at most
$2B(1-h_p(S^1))$, and the value of the first integral is at least
$B(1-h_p(S^1))$ because $\ell(S^{1}) \geq B$.  It follows that the
value of the second integral is at most twice the value of the first,
so we have
\[ \sum_{v \in Y}\tilde{p}(v)c_v \leq 3\int_{t=0}^{\ell(S^1)} 1 -
  h_p(S^1_{\langle t\rangle} )dt \]
which yields the following inequality, allowing us to bound the total
cost of $S^{Y}$ by analyzing the cost of $S^{1}$.

\begin{equation}
\label{eq:twotoone}
\cost(h_p,S^Y) \leq 3\cost(h_p,S^1)
\end{equation}

Therefore, to prove Lemma~\ref{lem:newlem}, it suffices to bound 
$\int_{t=0}^{\ell(S^1)} 1 - h_p(S^1_{\langle t\rangle} )dt$, which is the cost of schedule $S_1$
with respect to $h_p$.

Schedule $S^{1}$
selects items greedily with respect to $h_{p}$. 
However, we cannot apply
Theorem~\ref{thm:streeter} to bound the cost of $S_1$ in terms of the
cost of $S^*$, because only items of cost at most $B$ are considered
in greedily forming $S^1$, while items of cost greater than $B$ may be
included in $S^*$.

We will instead bound the cost of $S^1$ in terms of the cost of
the truncated schedule $S^*_{\langle B \rangle}$.
To do this, we will prove a lemma that is similar to Theorem~\ref{thm:streeter}.
We defer its proof to the next section, since
it is somewhat technical and is similar to
the proof of \theoremref{thm:streeter}.
The definitions of $G_j$ and $d$ are as given in
the previous section.

The statement of the lemma is as follows.
%Let $S^G$
%denote the schedule produced by running their greedy algorithm.
%Let $G_j$ denote the prefix of $S^G$ consisting of the first $j-1$ items
%in that schedule, together with their lengths.

\begin{lemma}
\label{lem:sgb}
Let $I$ be an instance of the Discrete Min-Sum Submodular Cover problem with
time vector $c$, utility function $f$, and item set $N$.
Let $\mathcal{S}$ denote the set of all schedules $S$ for item set $N$ and
cost vector $c$ satisfying $f^{c}(S) = f(N)$.
Let $G = \langle g_1, g_2, \dots \rangle$
be the schedule constructed by running Standard Greedy on instance $I$
and let
$G_j = \langle g_1, g_2, \dots g_{j-1}\rangle$, where
$G_1 = \langle\,\rangle$. 
Let $B \in \mathbb{R}$ be such that $\ell(G) \geq B$
and let $d$ be the maximum $j$ such that
$\ell(G_j) < B$.
For any
schedule $S \in \mathcal{S}$,
$\cost(f,G_d) \leq 4\cost(f,S_{\langle B \rangle})$.  Further,
$\cost(f,G_{d+1}) \leq 8\cost(f,S_{\langle B \rangle})$.
\end{lemma}

We now show how to use \lemmaref{lem:sgb} to prove Lemma~\ref{lem:newlem}.  

  Let $m$ be such that
  $S^*_{\langle B \rangle} = \langle (j_1, c_{j_1}), \ldots,
  (j_{m-1},c_{j_{m-1}}), (j_m, \tau_{j_m}) \rangle$.  By the definition of schedule
  truncation, $\tau_{j_m} \leq c_{j_{m}}$.  Since the length of
  $S^{*}_{\langle B \rangle}$ is $B$, each of
  $c_{j_1}, \ldots, c_{j_{m-1}}$ is at most $B$, but it is possible
  that $c_{j_m} > B$.

  Consider a restricted version $I'$ of our current Min-Sum Submodular
  Cover instance $I$ in which we include only those items $i \in N$
  such that $c_i \leq B$.  Let $N'$ be the set of those items.
% Let $G'$ be the schedule that results from running the greedy algorithm
% for Min-Sum Submodular Cover on $I'$.
  Let $S'$ be the schedule that results from concatenating
  $\langle (j_1, c_{j_1}), \ldots, (j_{m-1},c_{j_{m-1}}) \rangle$ with
  an arbitrary sequence of pairs $(i,c_i)$ with $i \in N'$, such that
  $h_{p}(S')=h_{p}(N')$.  Let $\ell'$ denote
  $\ell(\langle (j_1, c_{j_1}), \ldots, (j_{m-1},c_{j_{m-1}}
  \rangle)$.
  Comparing $S'_{\langle B \rangle}$ to $S^{*}_{\langle B \rangle}$,
  both have the same first $m-1$ elements.  Schedule
  $S^{*}_{\langle B \rangle}$ then has $(j_m, \tau_{j_{m}})$ where
  $\tau_{j_{m}} = B - \ell'$, whereas schedule
  $S'_{\langle B \rangle}$ may then have multiple elements in
  $N' \times \mathbb{Z}_{\geq 0}$ which together have length
  $B - \ell'$.  Because $h_p$ is monotone, and the cost of $h_p$ on
  schedule $S^{*}_{\langle B \rangle}$ is
  $\cost(h_{p}, S^{*}_{\langle B\rangle}) = \int_{t=0}^{B} 1 -
  h_p(S^{*}_{\langle t\rangle} )dt$,
  and analogously for $S'_{\langle B \rangle}$, it immediately follows
  that
\begin{equation}
\label{eq:primo}
\cost(h_p, S'_{\langle B \rangle}) \leq \cost(h_p,S^*_{\langle B \rangle})
\end{equation}

Now consider the schedule $S^1$ that is computed during the the first stage of
running \texttt{MixedGreedy}. 
Let $G'$ be the greedy schedule produced by
running the Greedy algorithm on instance $I'$,
with utility function $h_p$ and times $c$. Because only items $i$
with $c_i \leq B$ are considered when $S^1$ is constructed, and items
are chosen greedily with respect to $h_p$, $S^1$ is a prefix of $G'$.

Let $d$ be such that
$S^1 = \langle (i_1, c_{i_1}), \ldots, (i_d, c_{i_d}) \rangle$. Thus,
$S^{1} = G'_{d+1}$. in particular, we have that $\ell(S^{1}) \geq B$
and
$\ell(\langle (i_{1},c_{i_{1}}),\dots,(i_{d-1},c_{i_{d-1}})\rangle) <
B$.
%For $k \in \{d,d+1\}$, let $G^k$ denote $\langle (i_1, c_{i_1}), \ldots, (i_d, c_{i_k}) \rangle$.
It follows from~\eqref{eq:primo} and from  Lemma~\ref{lem:sgb} that
\begin{equation}
\label{eq:onetoprime}
\cost(h_p,S^1) \leq 8\cost(h_{p},S'_{\langle B\rangle}) \leq 8\cost(h_p,S^*_{\langle B \rangle})
\end{equation}
and therefore 
\begin{equation}
\cost(h_p,S^1) \leq 8\cost(h_p,S^*)
\end{equation}

We have that $\cost(h_p,S^Y) \leq 3\cost(h_p,S^1)$. We also have that
$c_Y = \cost(h_p,S^Y)$ and $C^* = \cost(h_p,S^*)$. Therefore, we have
\[
c_{Y} = \cost(h_{p}, S^{Y}) \leq 3\cost(h_{p}, S^{1}) \leq 24\cost(h_{p}, S^{*}) = 24C^{*}
\]
\jmlrQED

\subsection{Proof of Lemma~\ref{lem:sgb}, approximation bounds for truncated schedules}

We prove Lemma~\ref{lem:sgb}, which states that the following two properties hold:

\begin{enumerate}[\bfseries{Property }1{:}]
\item $\cost(f,G_d) \leq 4\cost(f,S_{\langle B \rangle})$
\item $\cost(f,G_{d+1}) \leq 8\cost(f,S_{\langle B \rangle})$
\end{enumerate}

The proof is similar to the proof of Streeter and Golovin for Theorem~\ref{thm:streeter}.
\footnote{Although we give a proof only for Discrete Min-Sum Submodular Cover, the proof can easily be adapted to give the same result for the more general Min-Sum Submodular Cover problem considered by Streeter and Golovin.}
We will assume that $f:2^N \rightarrow [0,1]$.
We can transform any $f:2^N \rightarrow \mathbb{R}_{\geq 0}$ into a function of this
type by scaling $f$ so that 
for all $S \in 2^N$, the scaled version of $f(S)$ is equal to $\frac{f(S) - f(\emptyset)}{f(N) - f(\emptyset)}$.

Recall that $f^c$ is the function defined on schedules $S$ such that
$f^c(S) = \frac{1}{f(N)}f(\{i\mid (i,c_{i}) \in S\})$.
We call $f^c$ a
\emph{job}.  
We refer to a pair
$(i,\tau) \in N \times \mathbb{R}_{\geq 0}$ as an \emph{action} and
to $\tau$ as the \emph{time} taken by that action.

As in Section~\ref{sec:newlemproof}, let 
$G = \langle (g_1,\tau_1), (g_2,\tau_2), \dots\rangle$,
denote the schedule
computed by the Greedy algorithm on $I$ and
let $G_j = \langle g_1, g_2, \dots g_{j-1}\rangle$.
et $S$ be an arbitrary
schedule for the instance with $f(S) = f(N)$.  
Let $d$ be the maximum $j$
such that $\ell(G_j) < B$.

We may assume without
loss of generality that for every $(i,\tau)$ in $S$,
$\tau = c_i$, 
since 
$f^c$ does not gain any value from pairs
$(i,\tau)$ with $\tau \neq c_i$.  
As before, we will generally omit the superscript on $f^c$ and simply write
$f(S)$. 

We begin by showing that Property 1 implies Property 2.

\noindent\textbf{Property 1 $\Rightarrow$ Property 2:}
  We define $f_{G_d}(S)$, a new function defined on schedules that is derived from $f$. 
  Intuitively, $G_d$ completes some portion of the job $f$; we wish to consider the portion
  of the job that remains to be completed after the actions in $G_{d}$
  have been performed. The function $f_{G_d}(S)$ is defined to be the
  portion of the job completed by first executing schedule $G_{d}$ and
  then executing schedule $S$. We express this as
  $f_{G_{d}}(S) = f(G_{d} \oplus S)$. Note that $f_{G_{d}}$ still
  satisfies the essential conditions for a job as it is monotone and
  submodular. It should be noted, however, that unless $f(G_{d})=0$,
  then $f_{G_{d}}(\langle\rangle) \neq 0$ (equivalently, due to
  monotonicity, there is no schedule $S$ for which $f_{G_{d}}(S)=0$).

  It is easy to show that the $\cost(f_{G_{d}}, S)$ represents the
  additional cost incurred by schedule $S$ on job $f$ after the
  schedule $G_{d}$ has already been executed.
  \begin{align*}
    \cost(f_{G_d}, S)
    &= \int_{t=0}^{\ell(S)} \left(1 - f_{G_d}(S_{\langle t \rangle})\right)dt \\
    &= \int_{t=0}^{\ell(S)} \left(1 - f(G_{d}\oplus S_{\langle t \rangle})\right)dt \\
    &= \int_{t=\ell(G_{d})}^{\ell(G_{d}\oplus S)} \left(1-f\left(\left(G_{d}\oplus S\right)_{\langle t \rangle}\right)\right)dt \\
    &= \int_{t=0}^{\ell(G_{d}\oplus S)} \left(1 - f\left(\left(G_{d}\oplus S\right)_{\langle t \rangle}\right)\right)dt - \int_{t=0}^{\ell(G_{d})} \left(1 - f(G_{d\,\langle t \rangle})\right)dt
  \end{align*}
  Therefore, we have
\begin{equation}\label{eq:costfgd}
  \cost(f, G_{d}) + \cost(f_{G_{d}}, S) = \cost(f, G_{d}\oplus S)
\end{equation}

Property~1 asserts that
$\cost(f, G_d) \leq 4\cost(f, S_{\langle B\rangle})$ for any schedule
$S \in \mathcal{S}$.  
STOPPED HERE
Also from this
assumption, the greedy schedule for $f_{G_{d}}$ is within a factor of
$4$ of any other schedule for $f_{G_{d}}$. Additionally, if we look at
only the first action of the greedy schedule for $f_{G_{d}}$
(i.e.~action $g_{d}$), the cost incurred by this one action is less than that
of the entire greedy schedule for $f_{G_{d}}$, which in turn is less
than $4$ times any other schedule for $f_{G_{d}}$. Thus, we also have
that
$\cost(f_{G_d}, \langle g_d\rangle) \leq 4\cost(f_{G_d},
S^{*}_{\langle B \rangle})$. Therefore, we have
\begin{align*}
 \cost(f, G_{d+1}) & = \cost(f, G_{d}) + \cost(f_{G_d}, \langle g_d\rangle)\tag{by \eqref{eq:costfgd}} \\
& \leq 4\cost(f, S^{*}_{\langle B \rangle}) + 4\cost(f_{G_d}, S^{*}_{\langle
      B\rangle})\\
&  \leq 8\cost(f, S^{*}_{\langle B \rangle})
\end{align*}
since, by the monotonicity of $f$, $\cost(f_{G_{d}}, S^{*}_{\langle B\rangle}) \leq \cost(f, S^{*}_{\langle B\rangle})$.
\bigskip

\noindent\textbf{Proof of Property 1:}
  We first define a few values. The quantity $R_j = 1 - f(G_j)$ represents how
  much of our task remains to be completed before the $j$th item of the greedy schedule is chosen.
  We define $s_j$ to be the ``bang for the buck''
  earned from that item. 
  That is, $s_j = \frac{(R_j - R_{j+1})}{\tau_j}$.
  Then, let
  $p_j = \frac{R_j}{s_j}$ for all $j \leq d$, and $p_j = 0$ for
  $j > d$. Let $x_j = \frac{p_j}{2}$ and let
  $y_j = \frac{R_j}{2}$. Also, let
  $\psi(x) = 1 - f(S^{*}_{\langle x\rangle})$.

  In order to prove the theorem, we wish to show
  \[ \int_{t=0}^B \left(1-f(S^{*}_{\langle t\rangle})\right)dt =
    \int_{x=0}^{B}\psi(x)\,dx \geq \frac{1}{4}\cost(f, G_{d}) \]

  We need to integrate $\psi(x)$ only up to $x=B$. When $x=B$,
  $\psi(x)=\psi(B)$ is the amount of the task that remains to be
  completed at time $B$ under schedule the optimal schedule
  $S^{*}$. We associate with this amount a $y_k$, corresponding to the
  greedy schedule, where $k=\min\{\,j : y_j \leq \psi(B)\,\}$. This
  can be seen in Figure~\ref{fig:graph}, where $x=B$ and $y=\psi(B)$
  are shown as dotted lines, with $y_{k}$ being the first $y_{j}$
  appearing below the dotted line $y=\psi(B)$.
  
  We first present an important fact. For any schedule $S$, any
  positive integer $j \leq d$, and any $t >= 0$,
  \begin{equation}\label{eq:sfact1} f(S_{\langle t\rangle})\leq
    f(G_j)+t\cdot s_j \end{equation}
  This is a consequence of the monotonicity and submodularity of $f$, together with the fact that the greedy algorithm 
  always chooses the item with the best ``bang for the buck''.
  It is shown in~\cite{streeterGolovin} as Fact~1.
   
  Using this fact, we have
  \[ f(S^{*}_{\langle x_j\rangle}) \leq f(G_j) + x_{j}s_j = f(G_j) +
    \frac{R_j}{2} \] So, for $j \leq d$ we have
  \begin{equation*}
    \psi(x_{j}) = 1 - f(S^{*}_{\langle x_{j} \rangle}) \geq 1 - f(G_{j}) - \frac{R_{j}}{2} = R_j - \frac{R_j}{2}
  \end{equation*}
and therefore
  
  \begin{equation}\label{eq:hxjyj}
    \psi(x_{j}) \geq y_{j}
  \end{equation}
  Note that \eqref{eq:hxjyj} holds for $j>d$ as well, since $x_{j}=0$
  by definition; thus, $\psi(x_{j})=\psi(0)=1 \geq y_{j}$.

  \begin{figure}[hp!]
    \floatconts{fig:graph}{\caption{Above $y_{k}$, the gray bars fit
        entirely inside the area of integration of $\psi-y_{k}$ (the area
        below the $\psi(x)$ graph and above $y_{k}$). Below $y_{k}$, the total area
        of the gray bars is still less than the remaining area of
        integration for $\psi$ (that is, the rectangle bounded above by
        $y_{k}$ and on the right by the dotted line $x=B$)}}{
%    \begin{subfigure}{.8\textwidth}
      \subfigure[The case where $d < k$. The area of the gray bars
      below $y_{k}$ is $0$.]{
        \begin{tikzpicture}
        \draw[color=gray] (0,0) -- (0,6) -- (9,6) -- (9,0) -- (0,0);
        \node at (1.25,4.75) {\smaller$\psi(x)$};
        \node at (4.5,-0.5) {\smaller$x$};
        \foreach \x in {0,...,6}
        {
          \draw[color=gray] (-0.05,\x) -- (0,\x);
          \draw[color=gray] (\x,-0.05) -- (\x,0);
        }
        % Gray parts
        \node at (-0.25,5) {\scriptsize$1$};
        \node at (-0.25,2.6) {\tiny$y_1$};
        \draw[fill,color=lightgray] (0,2.6) -- (1,2.6) -- (1,2.3) -- (0,2.3);
        \draw[fill,color=lightgray] (0,2.3) -- (2,2.3) -- (2,1.9) -- (0,1.9);
        \draw[fill,color=lightgray] (0,1.9) -- (1.8,1.9) -- (1.8,1.6) -- (0,1.6);
        \draw[thick] (0,2.6) -- (1,2.6);
        \node at (-0.25,2.3) {\tiny$y_2$};
        \draw[thick] (0,2.3) -- (2,2.3);
        \node at (-0.25,1.9) {\tiny$y_d$};
        \draw[thick] (0,1.9) -- (1.8,1.9);
        \node at (-0.25,1.6) {\tiny$y_k$};
        \draw[thick] (0,1.6) -- (0.05,1.6);
        % Budget Line
        \draw[dashed,thick] (0,1.8) -- (2.25,1.8);
        \draw[dashed,thick] (3.75,1.25) -- (3.75,0);
        \node at (3.75,-0.25) {\tiny$B$};
        % Blue Line
        \draw[color=blue,thick] (0,5) -- (0.75,5) -- (0.75,3) -- (2.25,3) -- (2.25,1.8) -- (3.75,1.8) -- (3.75,1.25) -- (5.25,1.25) -- (5.25,0.5) -- (7.5,0.5) -- (7.5,0);
        % Forced spacing?
        \node at (-1.5,0) {\ };
        \node at (9.5,0) {\ };
        \end{tikzpicture}
      }
%      \vspace{-8mm}
%      \caption{ The case where $d < k$. The area of the gray bars
%        below $y_{k}$ is $0$.}
%    \end{subfigure}

%    \begin{subfigure}{.8\textwidth}
      \subfigure[The case where $d > k$. The area of the gray bars
      below $y_{k}$ is nonzero. Note that the bars below $y_{k}$ may
      extend past $x=B$.]{
        \begin{tikzpicture}
        \draw[color=gray] (0,0) -- (0,6) -- (9,6) -- (9,0) -- (0,0);
        \node at (1.25,4.75) {\smaller$\psi(x)$};
        \node at (4.5,-0.5) {\smaller$x$};
        \foreach \x in {0,...,6}
        {
          \draw[color=gray] (-0.05,\x) -- (0,\x);
          \draw[color=gray] (1.5*\x,-0.05) -- (1.5*\x,0);
        }
        % Gray parts
        \node at (-0.25,5) {\scriptsize$1$};
        \node at (-0.25,2.6) {\tiny$y_1$};
        \draw[fill,color=lightgray] (0,2.6) -- (1,2.6) -- (1,2.3) -- (0,2.3);
        \draw[fill,color=lightgray] (0,2.3) -- (2,2.3) -- (2,1.9) -- (0,1.9);
        \draw[fill,color=lightgray] (0,1.9) -- (1.7,1.9) -- (1.7,1.6) -- (0,1.6);
        \draw[fill,color=lightgray] (0,1.6) -- (1.5,1.6) -- (1.5,1.2) -- (0,1.2);
        \draw[fill,color=lightgray] (0,1.2) -- (2.4,1.2) -- (2.4,0.85) -- (0,0.85);
        \draw[fill,color=lightgray] (0,0.85) -- (4.5,0.85) -- (4.5,0.5) -- (0,0.5);
        \draw[thick] (0,2.6) -- (1,2.6);
        \node at (-0.25,2.3) {\tiny$y_2$};
        \draw[thick] (0,2.3) -- (2,2.3);
        \node at (-0.25,1.9) {\tiny$y_3$};
        \draw[thick] (0,1.9) -- (1.7,1.9);
        \node at (-0.25,1.6) {\tiny$y_k$};
        \draw[thick] (0,1.6) -- (1.5,1.6);
        \node at (-0.25,1.2) {\tiny$y_5$};
        \draw[thick] (0,1.2) -- (2.4,1.2);
        \node at (-0.25,0.85) {\tiny$y_d$};
        \draw[thick] (0,0.85) -- (4.5,0.85);
        \draw[thick] (0,0.5) -- (0.05,0.5);
        % Budget Line
        \draw[dashed,thick] (0,1.8) -- (2.25,1.8);
        \draw[dashed,thick] (3.75,1.25) -- (3.75,0);
        \node at (3.75,-0.25) {\tiny$B$};
        % Blue Line
        \draw[color=blue,thick] (0,5) -- (0.75,5) -- (0.75,3) -- (2.25,3) -- (2.25,1.8) -- (3.75,1.8) -- (3.75,1.25) -- (5.25,1.25) -- (5.25,0.5) -- (7.5,0.5) -- (7.5,0);
        % Forced spacing?
        \node at (0,6.5) {\ };
        \node at (-1.5,0) {\ };
        \node at (9.5,0) {\ };
        \end{tikzpicture}
      }
    }
%      \vspace{-8mm}
%      \caption{ The case where $d > k$. The area of the gray bars
%        below $y_{k}$ is nonzero. Even though the bars below $y_{k}$
%        may extend past $x=B$ and fall outside the area of
%        integration, the area is still less than the area under $h$
%        and left of $x=B$. }
%    \end{subfigure}
  
%    \caption{ Above $y_{k}$, the gray bars fit entirely inside the
%      area of integration of $h-y_{k}$ (the area below the $\psi(x)$
%      graph, above $y_{k}$, and left of the dotted line representing
%      $x=B$). Below $y_{k}$, the total area of the gray bars is still
%      less than the remaining area of integration for $h$ (that is,
%      the rectangle bounded above by $y_{k}$ and to the right by the
%      dotted line $x=B$).}\label{fig:graph}
  \end{figure}
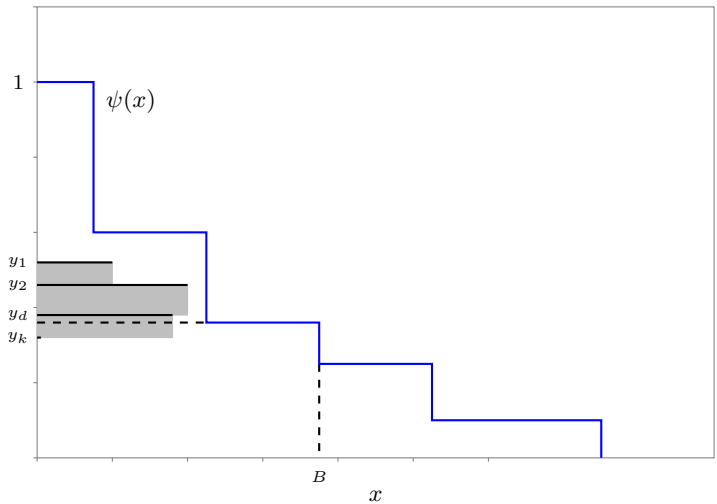
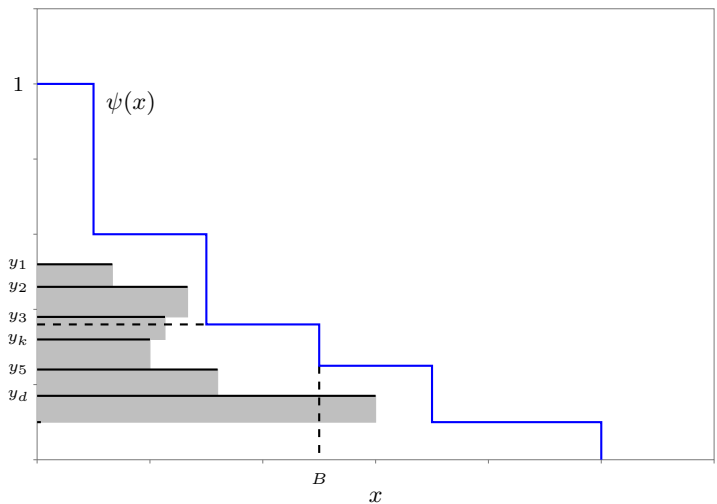
  
The cost of the greedy schedule is $\sum_{j=1}^d R_j \tau_j$.
The quantity $R_j{\tau_j}$ is the contribution of action $j$ to the cost of the greedy schedule.
We can think of this quantity as charging $\tau_j$ per unit of $R_j$.
We can rewrite the contribution by instead dividing the charge per unit of utility change,
$R_j - R_{j+1}$.  That is, we can rewrite $R_j {\tau_j}$ as the product of
  $R_{j}\tau_{j}/\left(R_{j}-R_{j+1}\right)$  and $R_j - R_{j+1}$.
It follows from the definitions that $R_{j}\tau_{j}/\left(R_{j}-R_{j+1}\right) = \frac{R_j}{s_j}$ and
therefore

\begin{equation}
\label{eq:relatetoxy}
x_j(y_j - y_{j+1}) = \frac{1}{4}{R_j}{\tau_j}
\end{equation}

Since $\cost(f,G_{d}) = \sum_{j=1}^d R_j\tau_j$, we now have

\begin{equation}
\label{eq:relatetoxy1}
\frac{1}{4}\cost(f, G_{d}) = \sum_{j=1}^{d} x_j (y_j-y_{j+1})
\end{equation}

The lemma now follows immediately from the following claim:

\begin{claim}
\label{claim:histo}
$\sum_{j=1}^{d} x_j (y_j-y_{j+1}) \leq \int_{x=0}^B \psi(x) dx$.
\end{claim}

To prove this claim, we note that for each $j$, we have a pair
$(x_{j}, y_{j})$. Figure~\ref{fig:graph} shows two histograms
(represented by gray bars). 
For any
given $j$, we have a gray bar such that the top of the bar is at
$y_{j}$, the bottom is at $y_{j+1}$, and the length of the bar is
$x_{j}$.

  Proving the claim is equivalent to showing
  that the total area of the gray bars does not
  exceed the integral of $\psi(x)$ up to $x=B$. 
  Combining 
  \eqref{eq:hxjyj} with the fact that $\psi$ is
  non-increasing, it follows that for a gray bar extending to a length
  of $x_{j}$, 
  the gray bar has a height no more than $y_{j}$ and thus
  is below the graph of $\psi$. This allows us to conclude that the gray
  bars fit entirely inside of the graph of $\psi$. However, since we are
  integrating $\psi$ only up to $x=B$, there may be some gray bars which,
  although they are within the graph of $\psi$, fall outside of the area
  of integration of $\psi$. These are the values $x_{j}$ such that $x_{j} >
  B$. We note the following important fact:
  \begin{fact}\label{fact:xjB}
    For all $j < k$, $x_j \leq B$. 
  \end{fact}
  The justification for this fact is as follows: 
  For any $x_{j} > B$, we
  know that $y_{j} \leq \psi(x_{j})$ from \eqref{eq:hxjyj} and
  $\psi(x_{j}) \leq \psi(B)$ since $\psi$ is nonincreasing. So, since
  $x_{j} > B$ implies that $y_{j} \leq \psi(B)$, we know that
  $y_{j} > \psi(B)$ implies that $x_{j} \leq B$. For all $j < k$, by the
  definition of $k$, we know that $y_{j} > \psi(B)$, and thus
  $x_{j} \leq B$.

  In order to show that the area of the histogram defined by the
  $(x_{j}, y_{j})$ pairs is no larger than the integral up to $x=B$ of
  $\psi(x)$, we will break the integral into two parts:
  \[ \int_{x=0}^B \psi(x)dx = \int_{x=0}^B (\psi(x) - y_k)dx +
    \int_{x=0}^B y_{k}dx \]
  and analyze each part. We note that the first part of the integral
  consists of the area above the line $y=y_{k}$. Above this line, the
  reasoning follows the same reasoning as in \cite{streeterGolovin}:
  Due to \eqref{eq:hxjyj}, we see that each bar is contained entirely
  inside the graph of $\psi$, and since $j < k$, the bar is entirely
  inside the area of integration.

  The second part of the integral consists of the
  area below $y=y_{k}$, where the bars are still inside the graph of
  $\psi$, but may extend past $x=B$ and thus fall outside the area of
  integration. We must use different reasoning to show that the area
  of the bars below $y=y_{k}$ do not exceed the area of $\psi$ below
  $y=y_{k}$ and left of $x=B$.

  Let
  $B' = \ell(G_{d})$. We have $B' = \sum_j \tau_j \geq \sum\limits_{j \geq k} \tau_j$. 
  % We note that $G_{d}$ takes a total time of $B'$. This fact leads
  % naturally to the observation that at any given point in $G_{d}$, the
  % amount of cost remaining to be incurred by $G_{d}$ cannot exceed $B'$ times the
  Therefore, using~\eqref{eq:relatetoxy}, and
the fact that $R_j \leq R_k$ for $j \geq k$,
% We use the alternate sum described
%   above to show the remaining cost incurred after step $k$:
  \begin{equation}\label{eq:ineq} \sum_{j=k}^d x_j(y_j -
    y_{j+1}) = \sum_{j = k}^d \frac{1}{4}\tau_j R_j \leq \frac{1}{4} R_k B'
\end{equation}

  This holds true even when $d < k$, as in this case the sum is simply $0$.

  Using this fact, combined with the fact that $\psi(x_j) \geq y_j$ for
  $j<d$, we can prove the claim. We have that
  \begin{equation}\label{eq:intineq}
    \int_{x=0}^B y_{k}dx = y_{k}B = \frac{1}{2}R_{k}B > \frac{1}{4}R_{k}B' \geq \sum_{j=k}^d x_j(y_j-y_{j+1})
  \end{equation}
  where the last inequality follows from~\eqref{eq:ineq}. If we look
  once again at Figure~\ref{fig:graph}, we see that for each $j$, we
  have a gray bar with area $x_{j}(y_{j}-y_{j+1})$. From
  \eqref{eq:hxjyj}, we know that the gray bars fit entirely inside
  $\psi(x)$, and so the area of the gray bars
  above $y_{k}$ is not more than the area under $\psi(x)$ and above
  $y_{k}$. That is,
  \begin{equation}\label{eq:ineqaboveyk}
    \int_{x=0}^B(\psi(x)-y_k)dx \geq \sum_{j=1}^{k-1}x_j\left(y_j-y_{j+1}\right)
  \end{equation}
  By using \eqref{eq:intineq} and \eqref{eq:ineqaboveyk}, we now have
  \begin{align}
    \int_{x=0}^{B} \psi(x)\,dx
    &= \int_{x=0}^B \left(\psi(x) - y_k\right)dx + \int_{x=0}^B y_{k}dx \notag\\
    &\geq \sum_{j=1}^{d} x_j\left(y_j-y_{j+1}\right)
  \end{align}
  as desired, thus proving Claim~\ref{claim:histo}. By proving
  Claim~\ref{claim:histo}, we have therefore also proven Property~1,
  and thus Lemma~\ref{lem:sgb}.\jmlrQED

\section{$O(k\log n)$-approximation for Scenario $k$-of-$n$ function evaluation}
\label{sec:rhofork}

% The general Stochastic Boolean Function Evaluation (SBFE) problem, as
% defined by Deshpande et al.~\cite{deshpandeHellersteinKletenik}, takes
% as input the representation of a Boolean function
% $f:\{0,1\}^n \rightarrow \{0,1\}$, probabilities $p_i$ for
% $i = 1, \ldots, n$ where $0 < p_i < 1$, and positive costs
% $c_i \geq 0$ for $i=1, \ldots, n$.  We need to determine the value of
% $f$ on an initially unknown input $a \in \{0,1\}^n$.  We can obtain
% the values of the bits of $a$ sequentially and adaptively, by paying
% cost $c_i$ to get the value of $a_i$.  We assume that $p_i$ is the
% probability that $a_i = 1$, and that the bits are independent.  We can
% stop paying for bits as soon as we know the value of $f(a)$.  The
% problem is to determine the adaptive order in which to request the
% bits, so as to minimize the expected value of the total cost paid,
% where expectation is with respect to the distribution defined by the
% $p_i$.
% 
% We define the Sample-Based version of this general problem by
% replacing the probabilities $p_i$ in the input with a weighted sample
% $S \subseteq \{0,1\}^n$, and minimizing the expected value with
% respect to the sample distribution.

Let $k \in \{0, \ldots, n\}$ and let $f:\{0,1\}^n \rightarrow \{0,1\}$
be the \textit{Boolean $k$-of-$n$ function} where $f(x) = 1$ iff at
least $k$ bits of $f$ are equal to 1.  To determine the value of this
$f$ on an unknown $a \in \{0,1\}^n$, we need to determine whether $f$
has at least $k$ ones, or at least $n-k+1$ zeros.  There is an elegant
polynomial-time exact algorithm solving the Stochastic BFE problem for Boolean
$k$-of-$n$ functions (cf.~\cite{Salloum79,Salloum84,BenDov81,Changetal90}).  

Here we consider the Scenario BFE problem for 
$k$-of-$n$ functions.
Following techniques used in a reduction of~\cite{deshpandeHellersteinKletenik}
for Stochastic BFE,
we reduce
this problem
to a Scenario SC problem, through the construction of an
appropriate utility function $g$ for the state set $\Gamma = \{0,1\}$.
We obtain $g$ by combining two other functions $g_0$, and $g_1$, with
respective goal values $n-k+1$ and $k$ respectively, using the standard OR
construction described in Section~\ref{sec:defs}.  Function
$g_1:\{0,1,*\}^n \rightarrow \mathbb{Z}_{\geq 0}$ is such that for all
$b \in \{0,1,*\}^n$, $g_1(b) = \min\{k,|\{i \mid b_i = 1\}|\}$.  Similarly,
$g_0(b) = \min\{n-k+1,|\{i \mid b_i = 0\}|\}$. Combining $g_0$ and $g_1$, and their goal
values using the OR construction yields the new function
$g:\{0,1,*\}^n \rightarrow \mathbb{Z}_{\geq 0}$ such that for
%$b \in \{0,1,*\}^n$, $g(b) = k(n-k+1) - (g_1(b)(n-k+1) + kg_0(b))$.
$b \in \{0,1,*\}^n$, $g(b) = k(n-k+1) - ((n-k+1)-g_{0}(b))(k-g_1(b))$.
The new goal value is $Q = k(n-k+1)$.  For $b \in \{0,1,*\}^n$,
$g(b) = Q$ iff $b$ either contains at least $(n-k+1)$ 0's or at least
$k$ 1's, and thus determining the value of $f$ on initially unknown
$a$ is equivalent to achieving goal value for $g$.

We now lower bound the value of parameter $\rho$ for this $g$.  For
$b \in \{0,1,*\}^n$ where $g(b) < Q$, and $i$ such that $b_i = *$,
$\Delta_g(b,i,1) \geq (n-k+1 - g_0(b))$ and
$\Delta_g(b,i,0) \geq (k - g_1(b))$.  Thus
$\frac{\Delta_g(b,i,1)}{Q-g(b)} \geq \frac{n-k+1 -
  g_0(b)}{(n-k+1-g_0(b))(k-g_1(b))} = \frac{1}{k-g_1(b)}$
and
$\frac{\Delta_g(b,i,0)}{Q-g(b)} \geq \frac{k - g_1(b)}{(n-k+1 -
  g_0(b))(k-g_1(b))} \geq \frac{1}{k}$.
The larger of these is at least $\frac{1}{k}$, and hence the value of
$\rho$ for $g$ is at least $\frac{1}{k}$.  It follows that running
Mixed Greedy on $g$ with respect to the
sample distribution, gives an $O(k \log n)$ approximation algorithm
for our Scenario Boolean $k$-of-$n$ function evaluation problem.
The bound $O(k \log n)$ has no dependence on the sample size or on the
weights.  For constant $k$, this bound is $O(\log n)$.
% We do not know whether there is a polynomial-time algorithm
% achieving an $O(\log n)$ approximation bound for non-constant $k$,
% or whether

Our Scenario $k$-of-$n$ function evaluation problem has
some similarities to the Generalized Min-Sum Set Cover problem, which
has a constant-factor approximation algorithm (see, e.g.,~\cite{skutellaWilliamson}).  However, in
the Generalized Min-Sum Set Cover problem, the goal is to find a {\em
  non-adaptive} strategy of minimum cost.  Further, the sample is
unweighted, and the covering requirements are different for
different assignments in the input sample.

% Deshpande et al. reduce evaluation problems for other Boolean
% functions $f$ to Submodular Cover problems, through the construction
% of other utility functions $g$, but the values of $\rho$ for these
% $g$ are typically linear in $n$.  Thus running Mixed Greedy on these
% $g$ results in an approximation factor that is linear in $n$, which
% is trivial.  We

\section{Adaptive Submodularity of $g_{W}$}\label{app:adaptivesubmod}
%\begin{proof}[Proof of Lemma~\ref{lem:gW}]
\begin{proof}\proofof{\lemmaref{lem:gW}}
Let $w(b) = \sum_{a\in S : a\succeq b}w(a)$ be the sum of
the weights of realizations in the sample $S$ that are extensions of
$b$. Then, we can write $h_{W}(b) = W - w(b)$.

The OR construction gives us
\begin{equation*}
  g_{W}(b) = QW - (Q-g(b))(W-h_{W}(b))
\end{equation*}

By the properties of the standard OR construction,
because $g$ and $h_W$ are monotone and submodular, so is
$g_{W}$. 

Let $b, b'\in(\Gamma\cup\{*\})^{n}$ such that $b'\succ b$, and
$i\in N$ where $b_{i}=b'_{i}=*$. To show that $g_{W}$ is adaptive submodular with respect to distribution $\mathcal{D}_{S,w}$, we must show that
$\mathbb{E}[\Delta{g_{W}}(b,i,\gamma)] \geq
\mathbb{E}[\Delta{g_{W}}(b',i,\gamma)]$ with respect to $\mathcal{D}_{S,w}$.

We start by finding $\Delta{g_{W}}(b,i,\gamma)$ for any
$b\in\left(\Gamma\cup \{*\}\right)^{n}$, $\gamma\in\Gamma$, and $i\in N$ such that
$b_{i}=*$:
\begin{equation*}
  \begin{split}
    \Delta{g_{W}}(b,i,\gamma)
    &= QW - \left(Q-g(b_{i\leftarrow \gamma})\right)\left(W-h_{W}(b_{i\leftarrow \gamma})\right) \\
    &\quad- QW + \left(Q -
      g(b)\right)\left(W-h_{W}(b)\right) \\
    &= Qh_{W}(b_{i\leftarrow \gamma}) + Wg(b_{i\leftarrow \gamma}) - g(b_{i\leftarrow \gamma})h_{W}(b_{i\leftarrow \gamma})\\
    &\quad- Qh_{W}(b) - Wg(b) + g(b)h_{W}(b) \\
    &= Q\Delta h_{W}(b,i,\gamma)+W\Delta g(b,i,\gamma)+g(b)h_{W}(b)-g(b_{i\leftarrow\gamma})h_{W}(b_{i\leftarrow\gamma})
  \end{split}
\end{equation*}
By adding and subtracting the same quantity, $g(b)h_{W}(b_{i\leftarrow\gamma})$, to the expression on the last line, we get
\begin{equation*}
  \begin{split}
    \Delta g_{W}(b,i,\gamma) &= Q\Delta h_{W}(b, i, \gamma) + W\Delta g(b, i, \gamma) + g(b)h_{W}(b) - g(b_{i\leftarrow\gamma})h_{W}(b_{i\leftarrow\gamma}) \\
    &\quad+ g(b)h_{W}(b_{i\leftarrow\gamma}) - g(b)h_{W}(b_{i\leftarrow\gamma}) \\
    &= Q\Delta h_{W}(b,i,\gamma) + W\Delta g(b,i,\gamma) \\
    &\quad-g(b)(h_{W}(b_{i\leftarrow\gamma})-h(b))-h_{W}(b_{i\leftarrow\gamma})(g(b_{i\leftarrow\gamma})-g(b))\\
    &= Q\Delta h_{W}(b,i,\gamma) + W\Delta g(b,i,\gamma)\\
    &\quad- \Delta h_{W}(b,i,\gamma)g(b) - \Delta g(b,i,\gamma) h_{W}(b_{i\leftarrow\gamma}) \\
    &= \Delta h_{W}(b,i,\gamma) (Q-g(b)) + \Delta
    g(b,i,\gamma)(W-h_{W}(b_{i\leftarrow\gamma}))
  \end{split}
\end{equation*}
We next recall that, by definition, $h_{W}(b) = W - w(b)$. Thus, we
have that $W - h_{W}(b_{i\leftarrow\gamma})=w(b_{i\leftarrow\gamma})$,
and we can simplify further:
\begin{equation*}
  \Delta g_{W}(b,i,\gamma) = \Delta h_{W}(b,i,\gamma)(Q-g(b)) + \Delta g(b,i,\gamma)w(b_{i\leftarrow\gamma})
\end{equation*}

We define for any partial realization $d$, the function
$\hat{Q}(d) = Q - g(d)$ to represent the amount of utility remaining to be
achieved by $d$. Also, let $U_{\gamma} = \Delta g(b,i,\gamma)$
represent the utility gained in $g$ by observing state $\gamma$ for
item $i$ in partial realization $b$. Let $W_{\gamma} = w(b_{i\leftarrow \gamma})$ represent the
weight of all realizations in $S$ consistent with $b_{i\leftarrow \gamma}$,
referred to as the \textit{total weight of state $\gamma$}. Let
$\overline{W}_{\gamma} = \sum_{\gamma' \neq \gamma} W_{\gamma'}$
represent the total weight of all states which are not $\gamma$. It is
clear that $\Delta h_{W}(b,i,\gamma) = \overline{W}_{\gamma}$. That
is, the change in utility in $h_{W}$ (or conceptually, the amount of
weight \textit{eliminated}) is equal to the total weight of states
which are not the observed state, $\gamma$. We can now substitute
these new values in the above equation and get:
\begin{equation*}
  \Delta {g_{W}}(b,i,\gamma) = \overline{W}_{\gamma}\,\hat{Q}(b) + U_{\gamma}W_{\gamma}
\end{equation*}

We now consider the calculation of the expected value of
$\Delta g_{W}$. For a realization $a\in\Gamma^{n}$ drawn from
$\mathcal{D}_{S,w}$, we have that
$Pr[a_{i}=\gamma \mid a\succeq b] =
\frac{w(b_{i\leftarrow\gamma})}{w(b)}=\frac{W_{\gamma}}{w(b)}$.
Then, the expected increase in utility is
\begin{align*}
  \mathbb{E}[\Delta {g_{W}} (b,i,\gamma)]
  &= \sum_{\gamma} \frac{W_{\gamma}}{w(b)}\Delta {g_{W}}(b,i,\gamma) \\
  &= \sum_{\gamma} \frac{W_{\gamma}}{w(b)}\left(\hat{Q}(b)\overline{W}_{\gamma}+U_{\gamma}W_{\gamma}\right) \\
  &= \frac{\sum_{\gamma} W_{\gamma}\overline{W}_{\gamma}\hat{Q}(b)+U_{\gamma}W^{2}_{\gamma}}{w(b)} \\
  &= \frac{\sum_{\gamma} W_{\gamma}\overline{W}_{\gamma}\hat{Q}(b)+U_{\gamma}W^{2}_{\gamma}}{\sum_{\gamma}W_{\gamma}}
\end{align*}
The last equality is true since $W_{\gamma}=w(b_{i\leftarrow\gamma})$
and the sum of $w(b_{i\leftarrow\gamma})$ for all $\gamma$ is equal to
$w(b)$.

We now consider the partial realization $b'$. The expected value on
partial realization $b'$ is analogous to the above expected value on
$b$:
\begin{equation*}
  \mathbb{E}[\Delta g_{W}(b',i,\gamma)]=\frac{\sum_{\gamma}W'_{\gamma}\overline{W}'_{\gamma}\hat{Q}(b')+U'_{\gamma}{W'}^{2}_{\gamma}}{\sum_{\gamma}W'_{\gamma}}
\end{equation*}
where $W'_{\gamma}=w(b'_{i\leftarrow\gamma})$,
$\overline{W}'_{\gamma}=\sum_{\gamma'\neq\gamma}W'_{\gamma'}$, and
$U'_{\gamma}=\Delta g(b',i,\gamma)$.

Next, let $\mathbf{W}=(W_{\gamma_{1}},W_{\gamma_{2}},\dots)$ be the
tuple containing all of the weights of the possible states with
respect to $b$, and let
$\mathbf{U}=(U_{\gamma_{1}},U_{\gamma_{2}},\dots)$ be the tuple
containing all of the $U_{\gamma}$ values for each of the possible
states. We also let
$\mathbf{W'}=(W'_{\gamma_{1}},W'_{\gamma_{2}},\dots)$ and
$\mathbf{U'}=(U'_{\gamma_{1}},U'_{\gamma_{2}},\dots)$.

It follows from the submodularity of $g$ that $\hat{Q}(b')\leq \hat{Q}(b)$ and
$U'_{\gamma}\leq U_{\gamma}$. Clearly $W'_{\gamma}\leq
W_{\gamma}$.
Finally, since $g$ is monotone, and the maximum value of $g$ on its
domain is $Q$, $U_{\gamma}\leq \hat{Q}(b)$ and $U'_{\gamma}\leq \hat{Q}(b')$.

Now let $r=|\Gamma|$. We will use
$w_{\gamma_{1}},w_{\gamma_{2}},\dots,w_{\gamma_{r}}$ to represent
variables for a new function which we will define. Similarly, we will
use $u_{\gamma_{1}},u_{\gamma_{2}},\dots,u_{\gamma_{r}}$ to represent
variables of the same function. We will also let
$\overline{w}_{\gamma}=\sum_{\gamma'\neq\gamma}w_{\gamma'}$ to
simplify the definition of the function. The $w_{\gamma}$ and
$u_{\gamma}$ variables are analogous to the $W_{\gamma}$ and
$U_{\gamma}$ in the expression for expected value above. We now define
our function $f\colon\mathbb{R}^{2|\Gamma|+1}\to\mathbb{R}$ such that
\begin{equation*}
  f(w_{\gamma_{1}},\dots,w_{\gamma_{r}},u_{\gamma_{1}},\dots,u_{\gamma_{r}},q)=\frac{\sum_{\gamma}q\overline{w}_{\gamma}w_{\gamma}+w^{2}_{\gamma}u_{\gamma}}{\sum_{\gamma}w_{\gamma}}
\end{equation*}

Note that this function is analogous to the formula for expected value
above. Specifically, we consider the point
$(\mathbf{W'},\mathbf{U'},\hat{Q}(b'))$ and the point
$(\mathbf{W},\mathbf{U},\hat{Q}(b))$. It should be noted that
$f(\mathbf{W'},\mathbf{U'},\hat{Q}(b')) = \mathbb{E}[\Delta
g_{W}(b,i,\gamma)]$
and
$f(\mathbf{W},\mathbf{U},\hat{Q}(b))=\mathbb{E}[\Delta
g_{W}(b',i,\gamma)]$.
Let $P$ be the path from the first point to the second point, which
increases $q$ continuously from $\hat{Q}(b')$ to $\hat{Q}(b)$, then
increases each $u_{\gamma_{i}}$ continuously from $U'_{\gamma_{i}}$ to
$U_{\gamma_{i}}$ for $i=1,2,\dots,r$, and finally increases each
$w_{\gamma_{i}}$ continuously from $W'_{\gamma_{i}}$ to
$W_{\gamma_{i}}$ for $i=1,2,\dots,r$. We show that for every point
along the path $P$, the partial derivatives of $f$ are non-negative,
and therefore the value of $f$ is nondecreasing along the path. This
proves that
$f(\mathbf{W},\mathbf{U},\hat{Q}(b))\geq
f(\mathbf{W'},\mathbf{U'},\hat{Q}(b'))$.
This implies that
$\mathbb{E}[\Delta g_{W}(b,i,\gamma)]\geq \mathbb{E}[\Delta
g_{W}(b',i,\gamma)]$,
and thus $g_{W}$ is adaptive submodular with respect to distribution
$\mathcal{D}_{S,w}$.

We let $K=\sum_{\gamma}w_{\gamma}$. Then, we start by taking the
partial derivative with respect to $q$:
\begin{equation*}
  \frac{\partial f}{\partial q} = \frac{\sum_{\gamma}\left(w_{\gamma}\overline{w}_{\gamma}\right)}{K} \geq 0
\end{equation*}
for all points on $P$ since all weights are nonnegative and thus $w_{\gamma}\geq 0$.

We also examine the partial derivative with respect to each
$u_{\gamma}$. Given any $\gamma$, the partial derivative is
\begin{equation*}
  \frac{\partial f}{\partial u_{\gamma}} = \frac{w_{\gamma}^{2}}{K} \geq 0
\end{equation*}
because $K$ is positive since all weights are nonnegative
(i.e.~$w_{\gamma} \geq 0$ for all $w_{\gamma}$).

Finally, for each $w_{\gamma}$, we will use the fact that
$\overline{w}_{\gamma}=\sum_{\gamma'\neq\gamma}w_{\gamma'}$. This
means that for any $\gamma$, we can express the sum of all weights as
$K=w_{\gamma}+\sum_{\gamma'\neq\gamma}w_{\gamma'}=w_{\gamma}+\overline{w}_{\gamma}$. This
fact is used several times in the following. We have
\begin{equation*}
  \begin{split}
    \frac{\partial f}{\partial w_{\gamma}} &= \frac{\left(\overline{w}_{\gamma}q+2w_{\gamma}u_{\gamma}+\sum\limits_{\gamma'\neq \gamma}w_{\gamma'}q\right)K - \left(\sum\limits_{\gamma'}\left[w_{\gamma'}\overline{w}_{\gamma'}q+w_{\gamma'}^{2}u_{\gamma'}\right]\right)}{K^{2}} \\
    &= \frac{\left(\overline{w}_{\gamma}q+2w_{\gamma}u_{\gamma}+\overline{w}_{\gamma}q\right)\left(w_{\gamma}+\overline{w}_{\gamma}\right) - \sum\limits_{\gamma'}\left(w_{\gamma'}\overline{w}_{\gamma'}q+w_{\gamma'}^{2}u_{\gamma'}\right)}{K^{2}} \\
    &=
    \frac{\left(2\overline{w}_{\gamma}w_{\gamma}q+2w_{\gamma}^{2}u_{\gamma}+2\overline{w}_{\gamma}^{2}q+2w_{\gamma}\overline{w}_{\gamma}u_{\gamma}\right)-\sum\limits_{\gamma'}\left(w_{\gamma'}\overline{w}_{\gamma'}q+w_{\gamma'}^{2}u_{\gamma'}\right)}{K^{2}}
  \end{split}
\end{equation*}
In the summation in the numerator, we look at the term for which $\gamma'=\gamma$ and we can
simplify the numerator:
\begin{equation*}
  \begin{split}
    \frac{\partial f}{\partial w_{\gamma}} &= \frac{w_{\gamma}\overline{w}_{\gamma}q+w_{\gamma}^{2}u_{\gamma}+2\overline{w}_{\gamma}^{2}q+2w_{\gamma}\overline{w}_{\gamma}u_{\gamma} - \sum\limits_{\gamma'\neq\gamma}\left(w_{\gamma'}\overline{w}_{\gamma'}q+w_{\gamma'}^{2}u_{\gamma'}\right)}{K^{2}}
  \end{split}
\end{equation*}
Then, we can find a lower bound on this expression for all points on
$P$. We note that initially, $u_{\gamma} \leq q$ since
$U_{\gamma}\leq \hat{Q}(b')$ for all $\gamma$. We first increase $q$
continuously to $\hat{Q}(b)$. Then we increase each $u_{\gamma}$
continuously from $U'_{\gamma}$ to $U_{\gamma}$. We also note that
$U'_{\gamma}\leq\hat{Q}(b)$, and so after we have increased each
$u_{\gamma}$ we still have that $u_{\gamma} \leq q$. So at all points
on the path we have that $u_{\gamma} \leq q$, and we can replace in the summation in the numerator each $u_{\gamma'}$ by $q$ to produce our lower bound:
\begin{equation*}
  \begin{split}
    \frac{\partial f}{\partial w_{\gamma}} &\geq \frac{w_{\gamma}\overline{w}_{\gamma}q+w_{\gamma}^{2}u_{\gamma}+2\overline{w}_{\gamma}^{2}q+2w_{\gamma}\overline{w}_{\gamma}u_{\gamma} - \sum\limits_{\gamma'\neq\gamma}\left(w_{\gamma'}\left(\overline{w}_{\gamma'}q+w_{\gamma'}q\right)\right)}{K^{2}} \\
    &= \frac{w_{\gamma}\overline{w}_{\gamma}q+w_{\gamma}^{2}u_{\gamma}+2\overline{w}_{\gamma}^{2}q+2w_{\gamma}\overline{w}_{\gamma}u_{\gamma} - \sum\limits_{\gamma'\neq\gamma}\left(qw_{\gamma'}\left(\overline{w}_{\gamma'}+w_{\gamma'}\right)\right)}{K^{2}} \\
    &= \frac{w_{\gamma}\overline{w}_{\gamma}q+w_{\gamma}^{2}u_{\gamma}+2\overline{w}_{\gamma}^{2}q+2w_{\gamma}\overline{w}_{\gamma}u_{\gamma} - q\sum\limits_{\gamma'\neq\gamma}\left(w_{\gamma'}K\right)}{K^{2}} \\
    &=
    \frac{w_{\gamma}\overline{w}_{\gamma}q+w_{\gamma}^{2}u_{\gamma}+2\overline{w}_{\gamma}^{2}q+2w_{\gamma}\overline{w}_{\gamma}u_{\gamma}
      -
      qK\sum\limits_{\gamma'\neq\gamma}\left(w_{\gamma'}\right)}{K^{2}}
  \end{split}
\end{equation*}
Then we note that, by definition,
$\overline{w}_{\gamma} = \sum_{\gamma'\neq\gamma}w_{\gamma'}$, and
simplify further:
\begin{equation*}
  \begin{split}
    &= \frac{w_{\gamma}\overline{w}_{\gamma}q+w_{\gamma}^{2}u_{\gamma}+2\overline{w}_{\gamma}^{2}q+2w_{\gamma}\overline{w}_{\gamma}u_{\gamma} - qK\overline{w}_{\gamma}}{K^{2}} \\
    &= \frac{\overline{w}_{\gamma}q\left(\overline{w}_{\gamma}+w_{\gamma}\right)+w_{\gamma}^{2}u_{\gamma}+\overline{w}_{\gamma}^{2}q+2w_{\gamma}\overline{w}_{\gamma}u_{\gamma} - qK\overline{w}_{\gamma}}{K^{2}} \\
    &= \frac{\overline{w}_{\gamma}qW+w_{\gamma}^{2}u_{\gamma}+\overline{w}_{\gamma}^{2}q+2w_{\gamma}\overline{w}_{\gamma}u_{\gamma} - qK\overline{w}_{\gamma}}{K^{2}} \\
    &= \frac{w_{\gamma}^{2}u_{\gamma}+\overline{w}_{\gamma}^{2}q+2w_{\gamma}\overline{w}_{\gamma}u_{\gamma}}{K^{2}} \\
    &\geq 0
  \end{split}
\end{equation*}
for all points on $P$ because $w_{\gamma}$ and $u_{\gamma}$ are nonnegative on $P$.

Thus, $f$ is nondecreasing along path $P$, and $g_{W}$ is adaptive
submodular with respect to the distribution $\mathcal{D}_{S,w}$.
\end{proof}
\end{document}